\documentclass[%
    aps,
    pra,
    reprint,
    amsmath,
    amssymb,
    floatfix
]{revtex4-2}

\usepackage{physics}               
\usepackage{mathtools}              
\usepackage{amsthm}                
\usepackage{graphicx}
\usepackage[caption=false]{subfig}  
\usepackage{bbm}                   

\theoremstyle{plain}
\newtheorem{theorem}{Theorem}
\newtheorem{Lemma}{Lemma}

\newcommand*{\targetPauli}[1]{\sigma_{#1}^{(T)}}
\newcommand*{\Id}[1][]{\mathbbm{1}_{#1}}

\begin{document}

\title{Defining Unique, Redundant, and Synergistic Quantum Information}

\author{Sean Ericson}
\email{sericson@uoregon.edu}
\affiliation{Department of Physics, University of Oregon, Eugene, OR 97403, USA}

\author{Hailin Wang}
\email{hailin@uoregon.edu}
\affiliation{Department of Physics, University of Oregon, Eugene, OR 97403, USA}

\author{S. J. van Enk}
\email{svanenk@uoregon.edu}
\affiliation{Department of Physics, University of Oregon, Eugene, OR 97403, USA}

\date{\today}

\begin{abstract}
We extend the classical ideas of the Partial Information Decomposition (PID) to the quantum domain and quantify unique, redundant, and synergistic quantum information. 
We show that unique information plays the central role in quantum error correction codes: any erasure-correctable subset of encoding qubits must contain zero unique information. Synergistic information between two disjoint subsets of encoding qubits appears when a logical operation is supported on the whole set but not on the subsets separately. 
In a different application of our PID, we show that redundant quantum information is the crucial feature of the decoherence mechanism proposed by Zurek \textit{et al.} to explain how the classical world emerges out of the quantum world.
\end{abstract}

\maketitle

\section{\label{sec:intro}Introduction}
The concept of redundant information is used in at least two areas of quantum mechanics without actually being defined precisely. 
First, it is used to describe colloquially how quantum error correction works. For example, one may explain that one quantum bit of information is encoded \textit{redundantly} in multiple physical qubits so as to be recoverable if just a single physical qubit is corrupted. 
Second, in order to explain how the objective classical world emerges out of a quantum-mechanical microscopic world, quantum Darwinism~\cite{Touil_2022} argues that it is the fact that many different parts of the environment share the same information about a quantum system that makes the latter appear classical. 
In the decoherence model used in~\cite{Touil_2022} all constituents of the environment start off in the same state and subsequently interact with the quantum system in the same way. 
Hence, whatever information is present is indeed redundant and one does not need a precise definition of what redundant information actually is or what other types of non-redundant information there may be.

In classical information theory, starting with the seminal work~\cite{williams2010} from 2010, a framework has been developed that allows one to define not just redundant information, but also the complementary concept, unique information, and synergistic information as well. 
The framework is known as Partial Information Decomposition (PID) and has been used extensively to characterize information flow in complex networks. 
For reviews, see, e.g.,\,~\cite{lizier2018,newman2022}.
The meaning of these concepts can be illustrated with a very simple example. Suppose we wish to learn a 3-bit string $T=(t_1,t_2,t_3)$ from two sets of bits $A=\{a_1,a_2\}$ and $B=\{b_1,b_2\}$ that are correlated in a known way with $T$, as follows
\begin{eqnarray}
    t_1 &=& a_1=b_1, \nonumber\\
    t_2 &=& b_2, \nonumber\\
    t_3 &=& a_2 \oplus b_2.
\end{eqnarray}
Then we can call the first bit of information, $t_1$, purely redundant, as we can obtain the information from $A$ alone but also from $B$ alone.
The second bit, $t_2$, we obtain uniquely from $B$. 
The third bit, $t_3$, cannot be obtained from either $A$ alone or from $B$ alone; we need both $A$ and $B$. That third bit is then called synergistic information.

It is important to note that in less straightforward examples multiple sensible definitions of unique, redundant, and synergistic information are possible, depending on what the goal is and what property one is interested in~\cite{bertschinger2014,finn2018,griffith2014,gutknecht2021,gutknecht2023,harder2013,ince2017,james2018,james2018b,kolchinsky2022,mediano2022,rauh2017secret,wibral2017a,wibral2017b,rosas2019,rosas2020,battiston2021,ay2019}. 
When trying to define a quantum version of the PID we should expect multiple sensible definitions to be possible as well.

This article is organized as follows.
We briefly present the main ideas behind the classical PID in Section~\ref{sec:PID}. 
Our main goal is to define a quantum version of the PID. 
We make a particular choice, based on superdense coding. 
We provide some reasons for this choice in Section~\ref{sec:capacity}, but the main justification follows from testing our definition in various applications.
In Section~\ref{sec:CPID} we define quantum unique information (the other quantities then follow automatically) in several steps and show that we get sensible results in simple test cases.
In fact, we argue that quantum unique information is the crucial concept (rather than redundant information) for quantum error correction: no physical qubit should contain any unique information, because losing it would prevent one from recovering the full original encoded information.
Sections~\ref{sec:QEC} and~\ref{sec:Decoherence} concern the detailed applications of our final definitions to the two cases mentioned above---quantum error correcting codes and an extension of the decoherence model
of~\cite{Touil_2022}---as well as the lessons we can learn about these two topics from our quantum PID. 
For example, the PID quantifies in what sense random codes approximate exact codes.
Those lessons are summarized in Section~\ref{sec:Summary}. 
The Appendix contains some more technical results pertaining to the application of our PID to particular classes of states (e.g.,\ pure states and quasi-classical states), as well as a proof of the non-negativity of the PID elements.
Additionally, we formulate and prove a particular theorem about the information capacity of subsets of stabilizer states.

\section{\label{sec:PID} Partial Information Decomposition}
Recent attention in the field of classical information theory has been given to the goal of extending the traditional Shannon quantities of entropy and mutual information to describe probability distributions over many variables.
While mutual information between subsets of variables is useful for classifying various distributions and elucidating relationships between the constituent variables, there are nonetheless many examples of different distributions of three or more variables that cannot be distinguished by mutual information alone~\cite{james2017}.
The Partial Information Decomposition was originally conceived in part to address this problem. 
Another motivation arises from the intuition---vague as it might be---that high-level information emerging in neural networks ought to be synergistic, in the sense that it is more than the sum of the individual neurons' information content.

\subsection{\label{sec:PID_classical}Classical PID}
The prototypical (and minimal) example of a PID considers a joint distribution of three variables $T$, $A$, and $B$. Here $T$ is the ``target'', which is the main variable of interest about which information is to be gained through measurement of the $A$ and $B$ variables. 
The idea is to decompose the mutual information $I(T;AB)$ between $T$ and $AB$ (considered jointly) into information about $T$ that is unique to $A$ (and similarly for $B$), information about $T$ that is redundant (in the sense that it could be learned by measuring either $A$ or $B$), and synergistic information that can only be learned by measuring $A$ and $B$ jointly:
\begin{eqnarray}
    I(T;AB) \; &=& I_{\text{unq}A} + I_{\text{unq}B} + I_{\text{red}} + I_{\text{syn}}, \nonumber \\
    I(T;A) \; &=& I_{\text{unq}A} + I_{\text{red}}, \nonumber \\
    I(T;B) \; &=& I_{\text{unq}B} + I_{\text{red}}.
\end{eqnarray}
A PID is determined by defining one of $I_\text{unq}$, $I_{\text{red}}$, or $I_{\text{syn}}$, as the remaining quantities are then fixed by the above three relations.

\subsection{\label{sec:PID_quantum}Quantum PID}
Just as classical information-theoretic quantities such as entropy, relative entropy, mutual information, etc. have been found to have corresponding versions in quantum information theory, it is natural to seek to define quantum unique, redundant, and synergistic information.
Indeed, well-known aspects of quantum information such as the no-cloning theorem and the monogamy of entanglement would seem to imply that certain multipartite quantum states necessarily possess some amount of unique information.

A definition for a partial information decomposition for quantum mutual information is given in~\cite{vanenk2023qpid} as a quantum generalization of a classical PID in which unique information is defined based on logarithmic pooling of the conditional distributions $P(T|A)$ and $P(T|B)$ as described in~\cite{van_Enk_2023}.
While this PID succeeds in differentiating states that are indistinguishable via the mutual information alone, its definition is purely formal.

As an alternative, we seek to construct an operational PID for the classical information capacity of the quantum system.
The classical information capacity of the target component of a multipartite state relative to other components is closely related to the mutual information.
Additionally, a particular aspect of the classical information capacity of quantum systems discussed below motivates the application of the partial information decomposition to this quantity.

\section{\label{sec:capacity}Classical Information Capacity}
Classical information capacity typically describes the ability of a quantum channel to transmit classical information~\cite{Hausladen_1996,schumacher1997}.
Quantum channels for the transmission of classical information consist of the composition of an encoding, transmission, and decoding channel.
For simplicity, we will consider only a fixed set of encoding channels and the noiseless (identity) transmission channel.
Details of the decoding channel will not be relevant for our considerations.
We will refer to the classical information capacity of the system itself, with the understanding that what is meant is the capacity of the considered channels as applied to the system.
Additionally, we will use the terms ``capacity'' and ``information'' interchangeably to refer to the classical information capacity.

\subsection{\label{sec:capacity_twoBit} Superdense Coding}
Superdense coding refers to the ability to use preexisting entanglement to send up to twice the amount of information as would classically be allowed by the dimensionality of the system (e.g.,\ a qubit in a maximally entangled state can be used to send two bits of classical information).
An important constraint on multi-party communication with quantum states is the so-called ``exclusion principle''~\cite{Prabhu_2013} for superdense coding.
For any multipartite quantum state, at most one pair of components can possess information capacity above the classical limit (i.e.,\ a \textit{quantum advantage}).
This fact motivates a partial information decomposition for superdense capacity, as a pair of components with a quantum advantage clearly possesses information capacity unique to that pair.

We consider the case of a tripartite state $\rho_{TAB}$, with the constituent components held by parties T, A, and B (referred to as the ``Target'', ``Alice'', and ``Bob'' for consistency with traditional quantum information nomenclature).
We focus on states for which the target component is a qubit (i.e.,\ the dimensionality of the target Hilbert space $\mathcal{H}_T$ is $d_T = 2$).
Alice and Bob's components, however, may be higher-dimensional and/or composite systems (i.e.,\ $d_A,\;d_B\geq2$).
Information is to be encoded via the application of one unitary from a predetermined set to the target qubit, before being sent through a noiseless quantum channel to either Alice or Bob, or perhaps both.
When the particular choice of unitary for any one encoding is unknown, the encoding process can be described by the application of a quantum channel to the target qubit (referred to as an \textit{encoding channel}) before it is sent via the noiseless quantum channel.

For a set of unitaries $\mathcal{U} = \{U_i\}$ with a priori selection probabilities $\{p_i\}$, we define the corresponding encoding channel by
\begin{equation}
    \Phi_{\text{encoding}}(\rho_{TAB}) = \sum_i p_i\left(U_i \otimes \mathbbm{1}_{AB}\right)\rho_{TAB}\left(U_i \otimes \mathbbm{1}_{AB}\right)^\dag.
\end{equation}
The joint information capacity $C_{AB}$ of such an encoding channel is given by the well-known \textit{Holevo bound},
\begin{equation} \label{eq:superdense_holevo_def}
    C_{AB} = \chi(\rho_{TAB};\Phi) \coloneqq H(\Phi(\rho_{TAB})) - H(\rho_{TAB}),
\end{equation}
where $H(\cdot)$ is the von Neumann entropy, and Alice and Bob's components are considered jointly.

For an arbitrary state, the maximum value of $C_{AB}$ is given by~\cite{Bowen_2001}
\begin{eqnarray}
    C_{AB} \; &=& \log d_T + H(\rho_{AB}) - H(\rho_{TAB}) \\
    &=& I(T;AB) + \log d_T - H(\rho_T).
\end{eqnarray}
The information from the target that Alice and Bob can individually receive (denoted $C_A$ and $C_B$, respectively) is calculated identically as in~\eqref{eq:superdense_holevo_def} with the irrelevant subsystem first traced out.

When the target subsystem is a qubit, this optimum is achieved by the application of the so-called \textit{completely depolarizing channel} to the target.
The completely depolarizing channel is defined by its action of mapping all input states to the maximally mixed state $\frac{1}{d}\mathbbm{1}$, where $d$ is the dimension of the Hilbert space of the state.
For qubits, this corresponds to the set of unitaries consisting of the identity and the Pauli matrices,
\begin{eqnarray} \label{eq:pauli_set}
    \mathcal{U} = \{\mathbbm{1}_T, \targetPauli{x}, \targetPauli{y}, \targetPauli{z}\}, \nonumber \\
    p_i = \frac{1}{4} \quad (i=1,2,3,4),
\end{eqnarray}
where we denote $\sigma_x \otimes \mathbbm{1}_{AB}$ by $\targetPauli{x}$, etc.

Note that, because information quantities must be invariant under local unitary operations, we can assume that the reduced density operator of the target qubit is diagonal in the standard basis.
This is equivalent (up to ordering and a phase) to defining the Pauli operators with respect to the eigenvectors of the reduced state of the target.

\subsection{\label{sec:capacity_oneBit} ``Normally''-dense Coding}
As the exclusion principle precludes the simultaneous transmission of more than one bit to Alice and Bob individually, we consider also a ``one-bit'' or ``normally-dense'' encoding process by which the target can send at most one bit of information.

A one-bit encoding can be achieved by any two-element subset of the unitary set given in~\eqref{eq:pauli_set}, and on average is optimized by using either unitary with equal probability.
Without loss of generality, however, we can assume that one of the two unitaries is the identity.
We therefore consider the three one-bit encoding channels
\begin{eqnarray}
    \Phi_X(\rho) \; &=& \frac{1}{2}\left(\rho + \targetPauli{x}\rho\targetPauli{x}\right), \nonumber \\
    \Phi_Y(\rho) \; &=& \frac{1}{2}\left(\rho + \targetPauli{y}\rho\targetPauli{y}\right), \nonumber \\
    \Phi_Z(\rho) \; &=& \frac{1}{2}\left(\rho + \targetPauli{z}\rho\targetPauli{z}\right),
\end{eqnarray}
where $\rho$ is either $\rho_{TA}$ or $\rho_{TB}$.

Note that these one-bit encoding channels are equivalent to the application of the \textit{completely dephasing channel} to the target qubit with respect to the eigenbasis of the respective Pauli matrix.
The completely dephasing channel with respect to a basis $\mathcal{B}$ has the effect of diagonalizing its input in that basis.

We will use the notation $X_A$ to refer to the capacity of the channel $\Phi_X$ as applied to $\rho_{TA}$, $X_B$ as applied to $\rho_{TB}$, and similarly for the channels $\Phi_Y$ and $\Phi_Z$.

For the purpose of one-bit encoding, defining the Pauli operators with respect to eigenvectors of $\rho_T$ leaves an ambiguity in the relative phase in superpositions of the form
\begin{equation}
    a\ket{0} + be^{i\phi}\ket{1}.
\end{equation}
To define a PID for superdense capacities, the use of quantities that are invariant with respect to this phase ambiguity is required.
To that end, we define the ``parallel'' and ``transverse'' one-bit capacities,
\begin{eqnarray}
    \mathcal{C}_{A(B)}^\parallel &=& Z_{A(B)}, \\
    \mathcal{C}_{A(B)}^\perp &=& \langle X_{A(B)} \rangle_\phi = \langle Y_{A(B)} \rangle_\phi.
\end{eqnarray}
That is, the parallel one-bit capacity is simply the $Z$ capacity, and the transverse capacity is the phase-averaged $X$ (or $Y$) capacity.

\section{\label{sec:CPID} Partial Information Decomposition for $C_{AB}$}

The collective two-bit capacity, being an information measure, should in theory admit a decomposition into capacity that is unique to Alice/Bob, capacity that is redundantly shared, and synergistic capacity that can only be realized by considering their components jointly.
Analogously to the decomposition of mutual information, we define the C-PID by
\begin{eqnarray}
    \label{eq:CPID_def}
    C_{AB} \; &=& C_{\text{unq}A} + C_{\text{unq}B} + C_{\text{red}} + C_{\text{syn}}, \nonumber \\
    C_A \; &=& C_{\text{unq}A} + C_{\text{red}}, \nonumber \\
    C_B \; &=& C_{\text{unq}B} + C_{\text{red}}.
\end{eqnarray}

We will first construct our proposed definition of unique information in Section~\ref{sec:CPID_def}. (The other two types of information are then uniquely fixed by~\eqref{eq:CPID_def}.) We then test and verify our definition on a few simple examples with just enough symmetry such that intuition unambiguously indicates when unique information ought to be nonzero.

\subsection{\label{sec:CPID_def} Definition}
We begin by defining one-bit unique capacities
\begin{eqnarray}
    \Gamma_A \; &=& \max\left(\mathcal{C}_A^\parallel - \mathcal{C}_B^\parallel, 0\right) + \max\left(\mathcal{C}_A^\perp - \mathcal{C}_B^\perp, 0\right), \nonumber \\
    \Gamma_B \; &=& \max\left(\mathcal{C}_B^\parallel - \mathcal{C}_A^\parallel, 0\right) + \max\left(\mathcal{C}_B^\perp - \mathcal{C}_A^\perp, 0\right).
\end{eqnarray}
These definitions for one-bit unique capacity follow from the intuition that either party may possess an advantage in either one-bit $Z$ encoding, the orthogonal encodings ($X$,$Y$), or both.

For pure states $\rho_{TAB} = \dyad{\psi_{TAB}}$ the one-bit unique capacities are in fact suitable definitions for unique superdense capacity, i.e.,\ $C_{\text{unq}A(B)} = \Gamma_{A(B)}$ (see Appendix~\ref{sec:PureStates}).
This is also the case for states that are completely decohered in a given basis $\ket{t,a,b}$,
\begin{eqnarray}
    \rho_{TAB} = \sum_{t,a,b}p_{t,a,b}\dyad{t,a,b},
\end{eqnarray}
in which case the parallel capacities are identically zero ($\mathcal{C}_A^\parallel = \mathcal{C}_B^\parallel = 0$), the transverse capacities and superdense capacities coincide ($\mathcal{C}_A^\perp = C_A$, $\mathcal{C}_B^\perp = C_B$), and the unique superdense capacity reduces to
\begin{eqnarray}
    C_{\text{unq}A} &=& \max(C_A - C_B, 0) \\
    C_{\text{unq}B} &=& \max(C_B - C_A, 0).
\end{eqnarray}

This is in fact equivalent (up to a shift in redundant information) to the well-known Minimal Mutual Information (MMI) PID (see Appendix~\ref{sec:ClassicalStates}).
An often-cited criticism of the MMI PID is its inability to differentiate ``amount'' of information from ``kind'' of information~\cite{harder2013, bertschinger2014}; whichever party has less information always has zero unique information, hence both parties cannot simultaneously have unique information.
Our quantum definition overcomes this shortcoming by defining unique capacity in terms of two distinct types of information (namely, the parallel and transverse capacities), thus allowing for the possibility of simultaneous nonzero unique capacity for Alice and Bob (see Section~\ref{sec:PMB}).

For a general mixed $\rho_{TAB}$, the consistency condition for unique capacity implied by~\eqref{eq:CPID_def},
\begin{equation}
    C_{\text{unq}A} - C_{\text{unq}B} = C_A - C_B,
\end{equation}
will not necessarily be satisfied under the identification $C_{\text{unq}A(B)} = \Gamma_{A(B)}$.
To ensure that this condition is satisfied we define the ``deficiency'' terms
\begin{eqnarray}
    \Delta_A &=& \left(C_A - C_B\right) - \left(\Gamma_A - \Gamma_B\right), \\
    \Delta_B &=& \left(C_B - C_A\right) - \left(\Gamma_B - \Gamma_A\right),
\end{eqnarray}
and the ``corrected'' one-bit unique capacities 
\begin{eqnarray}
    \gamma_A &=& \Gamma_A + \min(\Delta_A, 0), \\
    \gamma_B &=& \Gamma_B + \min(\Delta_B, 0),
\end{eqnarray}
which account for the non-trivial relationship between one-bit and two-bit capacities in the case of a mixed state.

We can then define unique two-bit capacity via
\begin{eqnarray}
    C_{\text{unq}A} &=&\max(\gamma_A, 0) - \min(\gamma_B, 0), \label{eq:C_unqA_def} \\
    C_{\text{unq}B} &=&\max(\gamma_B, 0) - \min(\gamma_A, 0), \label{eq:C_unqB_def}
\end{eqnarray}
with expressions for $C_\text{red}$ and $C_\text{syn}$ then following from the relations in~\eqref{eq:CPID_def}:
\begin{eqnarray}
    C_\text{red} &=& \frac{1}{2}\left(C_A + C_B - \abs{\gamma_A} - \abs{\gamma_B}\right), \label{eq:C_red_def}\\
    C_\text{syn} &=& C_{AB} - \frac{1}{2}\left(C_A + C_B + \abs{\gamma_A} + \abs{\gamma_B}\right). \label{eq:C_syn_def}
\end{eqnarray}

This definition of a C-PID contains four cases, determined by the sign and magnitude of $C_A - C_B$ with respect to $\Gamma_A$ and $\Gamma_B$. 
The four cases are listed in Table~\ref{tab:CPID}.
When broken into these four cases, the corrected one-bit unique capacities $\gamma_A$ and $\gamma_B$ can be eliminated entirely in favor of the two-bit capacities and the uncorrected one-bit unique capacities.
While equations~\eqref{eq:C_unqA_def}-\eqref{eq:C_syn_def} may appear somewhat convoluted, their corresponding expressions in each case are markedly simpler.

Additionally, from~\eqref{eq:C_unqA_def} and~\eqref{eq:C_unqB_def} it is clear that the unique capacity is non-negative. 
That the redundant and synergistic capacities are also non-negative is less obvious, but is proven in Appendix~\ref{sec:Non-Negativity}.

\begin{table*}
    \caption{
        \label{tab:CPID} 
        The four cases of the C-PID.
        Depending on the relative values of the difference between the individual superdense capacities ($C_A - C_B$) and the uncorrected one-bit unique capacities ($\Gamma_A,\;\Gamma_B$) the PID quantities $C_{\text{unq}A(B)}$, $C_\text{red}$ and $C_\text{syn}$ take on significantly simpler forms than suggested by equations~\eqref{eq:C_unqA_def}-\eqref{eq:C_syn_def}.
        Note that the four cases come in two complementary pairs: cases 1 and 4 are symmetric under the exchange of labels $A \leftrightarrow B$ (and similarly for cases 2 and 3).
        Only in cases 2 and 3 is it possible for both $C_{\text{unq}A}$ and $C_{\text{unq}B}$ to be nonzero.
        When $\rho_{TAB}$ is a pure state, case 1 becomes degenerate with case 3, and similarly case 2 becomes degenerate with case 4.
    }
    \begin{ruledtabular}
        \begin{tabular}{c|ccccc}
            Case & Condition & $C_{\text{unq}A}$ & $C_{\text{unq}B}$ & $C_\text{red}$ & $C_\text{syn}$ \\
            \hline
            1 & $C_A - C_B \leq -\Gamma_B$ & 0 & $C_B - C_A$ & $C_A$ & $C_{AB} - C_B$ \\
            2 & $-\Gamma_B < C_A - C_B \leq \Gamma_A - \Gamma_B$ & $C_A - C_B + \Gamma_B$ & $\Gamma_B$ & $C_B - \Gamma_B$ & $C_{AB} - C_A - \Gamma_B$ \\
            3 & $\Gamma_A - \Gamma_B < C_A - C_B \leq \Gamma_A$ & $\Gamma_A$ & $C_B - C_A + \Gamma_A$ & $C_A-\Gamma_A$ & $C_{AB} - C_B - \Gamma_A$ \\
            4 & $\Gamma_A < C_A - C_B$ & $C_A - C_B$ & 0 & $C_B$ & $C_{AB} - C_A$ \\
          \end{tabular}
    \end{ruledtabular}
\end{table*}

\subsection{\label{sec:CPID_ex} Examples}
Here we apply the proposed C-PID to some well known tripartite entangled states whose information decompositions are intuitively obvious a priori.

\subsubsection{The GHZ and W States}
Two famously inequivalent~\cite{D_r_2000} tripartite entangled states are the GHZ and W states, given by
\begin{equation}
    \ket{\text{GHZ}} = \frac{1}{\sqrt{2}}\left(\ket{000} + \ket{111}\right),
\end{equation}
and
\begin{equation}
    \ket{\text{W}} = \frac{1}{\sqrt{3}}\left(\ket{001} + \ket{010} + \ket{100}\right).
\end{equation}
Both of these states are symmetric between all three components, so the assignment of components to $T$, $A$, and $B$ is arbitrary.
Additionally, as a consequence of the symmetry, we should expect neither Alice nor Bob to possess unique information.
Indeed, we find that in both cases the individual parallel and transverse one-bit capacities are identical between Alice and Bob, giving zero unique superdense coding capacity.
The redundant capacity for both states is one bit, and the remainder of the collective superdense capacity is then synergistic.
For the GHZ state, in the context of quantum error correction, the one redundant bit can be interpreted as being purely classical (see Section~\ref{sssec:Repetition_Codes}).
Note also that, for these states, the $X$ and $Y$ capacities are invariant with respect to the phase described in Section~\ref{sec:capacity_oneBit}, thus averaging is not required.
The various capacities and resulting C-PIDs for these states are listed in Table~\ref{tab:W_and_GHZ_CPID}.

\begin{table}
    \caption{
        \label{tab:W_and_GHZ_CPID}
        Two-bit capacities, one-bit capacities, and C-PID values for the GHZ and W states in terms of the Shannon entropy function~\eqref{eq:target_entropy}.
        Due to the symmetry of these states, individual capacities for $A$ and $B$ are identical, hence only $A$ capacities are listed.
        As expected, unique information is absent in both states and both Alice and Bob's individual capacities are completely redundant.
    }
    \begin{ruledtabular}
        \begin{tabular}{c|cc|cc|ccc}
            & $C_{AB}$ & $C_A$ & $\mathcal{C}_A^\parallel$ & $\mathcal{C}_A^\perp$ & $C_{\text{unq}A}$ & $C_\text{red}$ & $C_\text{syn}$ \\
            \hline
            GHZ & 2 & 1 & 0 & 1 & 0 & 1 & 1 \\
            W & $1+h(\frac{1}{3})$ & 1 & $\frac{2}{3}$ & $1 + h(\frac{3 + \sqrt{5}}{6}) - h(\frac{1}{3})$ & 0 & 1 & $h(\frac{1}{3})$
        \end{tabular}
    \end{ruledtabular}
\end{table}

\subsubsection{Parameterized Mixed-Bell States} \label{sec:PMB}
Consider the states
\begin{eqnarray}
    \ket{\Psi_A} &=& \ket{\Phi^+}_{TA}\otimes\ket{\psi}_B, \\
    \ket{\Psi_B} &=& \ket{\Phi^+}_{TB}\otimes\ket{\phi}_A,
\end{eqnarray}
where 
\begin{equation}
    \ket{\Phi^+} = \frac{1}{\sqrt{2}}\left(\ket{00} + \ket{11}\right)
\end{equation}
is a Bell state and $\ket{\psi}$ and $\ket{\phi}$ are arbitrary single-qubit ``spectator'' states.
Consider further the convex combination of these states, which we will refer to as a ``mixed-Bell'' state:
\begin{equation}
    \rho_\text{MB}(\ket{\psi}, \ket{\phi}; p) = (1-p)\dyad{\Psi_A} + p\dyad{\Psi_B}.
\end{equation}
The parameterized state $\rho_\text{MB}(\ket{\psi}, \ket{\phi}; p)$ interpolates between a state in which the target is maximally entangled with Alice ($p=0$) and one in which the target is maximally entangled with Bob ($p=1$).
For all intermediate values of $p$ the state is mixed, with the purity being minimized at $p=1/2$.
Superdense capacities for two cases of this state (one with $\ket{\psi} = \ket{\phi} = \ket{0}$ and one with $\ket{\psi} = \ket{+}, \ket{\phi} = \ket{0}$) are plotted in Figure~\ref{fig:MixedBellCapacities}. 
\begin{figure}
    \centering
    \includegraphics[width=\columnwidth]{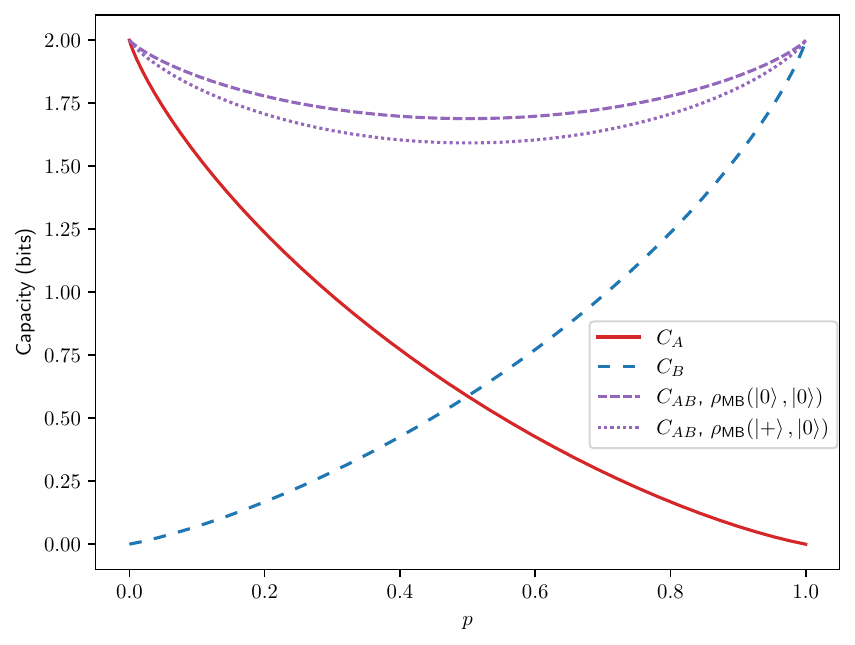}
    \caption{
        Collective and individual two-bit capacities for the states $\rho_\text{MB}(\ket{0}, \ket{0}; p)$ and $\rho_\text{MB}(\ket{+}, \ket{0}; p)$.
        For $p=0,1$, the target is maximally entangled with one of the parties, hence the joint capacity is maximized and equal to that party's individual capacity.
        The individual capacities are independent of $\ket{\psi}$ and $\ket{\phi}$, although different choices for these ``spectator'' states do slightly affect the joint capacity.
    }
    \label{fig:MixedBellCapacities}
\end{figure}
Note that the individual capacities $C_A$ and $C_B$ are independent of $\ket{\psi}$ and $\ket{\phi}$.

The C-PIDs for these two states are plotted in Figures~\ref{fig:MixedBellCPID_1} and~\ref{fig:MixedBellCPID_2}, respectively.
\begin{figure*}
    \subfloat[\label{fig:MixedBellCPID_1}]{%
        \includegraphics[width=\columnwidth]{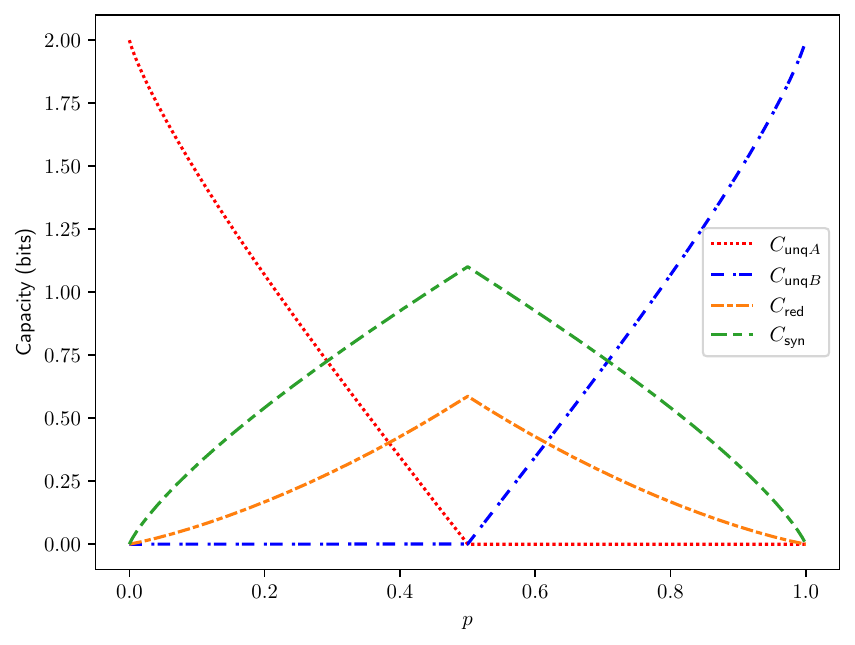}%
    }
    \hspace*{\fill}%
    \subfloat[\label{fig:MixedBellCPID_2}]{%
        \includegraphics[width=\columnwidth]{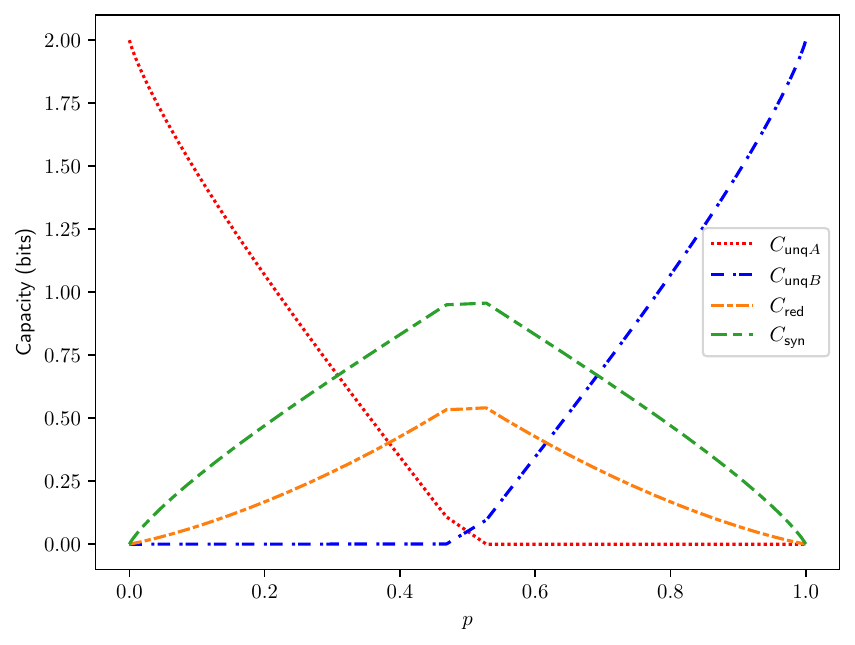}%
    }
    \caption{
        C-PID for two instances of the mixed-Bell state. 
        In a) the spectator states are $\ket{\psi} = \ket{\phi} = \ket{0}$.
        The symmetry of this state precludes the possibility of Alice and Bob simultaneously having unique capacity.
        In b) the spectator states are $\ket{\psi} = \ket{+}$ and $\ket{\phi} = \ket{0}$.
        The asymmetry in this case allows for a range of $p$ values ($0.47 \lesssim p \lesssim 0.52$) over which \textit{both} Alice and Bob have nonzero unique capacity.
        Note that in this case the state is not invariant under $A\leftrightarrow B$ and $p \leftrightarrow 1-p$, hence the slight asymmetry in the plot.
    }
\end{figure*}
At the extremes ($p=0,1$), the maximally entangled party has two bits of capacity. 
Indeed, this is entirely unique capacity, as the other party is completely independent from the target.

In the symmetric case ($\ket{\psi} = \ket{\phi} = \ket{0}$), the state is invariant under $p \leftrightarrow 1-p$ and $A \leftrightarrow B$.
Consequently, only the party that is more strongly entangled with the target has nonzero unique capacity.
For the asymmetric case there exists a range of values of $p$ for which Bob's parallel one-bit capacity is greater than Alice's ($\mathcal{C}_B^\parallel > \mathcal{C}_A^\parallel$), while Alice's one-bit transverse capacity is greater than Bob's ($\mathcal{C}_A^\perp > \mathcal{C}_B^\perp$).
In this region, then, both Alice and Bob have nonzero unique capacity.

In both cases the synergistic capacity is strictly greater than the redundant capacity for all $0<p<1$.
This is to be expected from considering the joint and individual superdense capacities, as $C_{AB} > C_A + C_B$ in this range.
From the defining relations of the C-PID~\eqref{eq:CPID_def}, we see that
\begin{equation}
    C_\text{syn} - C_\text{red} = C_{AB} - (C_A + C_B),
\end{equation}
thus for $p \neq 0,1$ we must have $C_\text{syn} > C_\text{red}$.

\section{\label{sec:QEC} Application: Stabilizer Quantum Error Correction Codes}
Like their classical counterparts, quantum error correction (QEC) codes are said to protect quantum information by ``redundantly'' encoding logical qubits into a larger number of physical qubits. 
Stabilizer QEC codes are a particular class of codes in which the set of possible codewords is defined as the simultaneous +1 eigenspace of an abelian subgroup (the ``stabilizer'') of the $n$-qubit Pauli group.
The stabilizer framework allows compact specification of codes, direct determination of their error correction and detection capabilities, and the construction of encoding, decoding, and syndrome measurement circuits.

\subsection{The Model}
To investigate redundant and unique information in stabilizer QEC codes, we consider a simple model (based on the model analyzed in~\cite{Cerf_1997}) in which one of a pair of qubits in a pure state is encoded into a logical qubit consisting of $n$ physical qubits according to various $[[n, 1, d]]$ stabilizer QEC codes.
Consider states of the form
\begin{eqnarray} \label{eq:psi_0}
    \ket{\psi_0} &=& \sqrt{\lambda}\ket{00} + \sqrt{1-\lambda}\ket{11} \\
    \to \ket{\psi_\text{enc}} &=& \sqrt{\lambda}\ket{00_\text{L}} + \sqrt{1-\lambda}\ket{11_\text{L}} \label{eq:encoded_state},
\end{eqnarray} 
where $\ket{0_\text{L}}$ and $\ket{1_\text{L}}$ are logical qubits, $\lambda$ and $1-\lambda$ are the Schmidt coefficients, and $0\leq \lambda\leq 1/2$.
Note that, while $\ket{\psi_0}$ as written in~\eqref{eq:psi_0} is defined on a 2-dimensional subspace of the full 4-dimensional 2-qubit Hilbert space, this is without loss of generality as the PID defined previously is invariant under local unitary transformations.

The unencoded half of the initial two-qubit state is designated as the target, while the physical qubits comprising the encoded logical qubit are numbered $1\dots n$ and constitute the $A$ and $B$ qubits for the purposes of the PID analysis.
One can then consider how information about the target qubit is distributed among various subsets of the physical qubits.

The mutual information and superdense capacity between the target and the initial (unencoded) qubit are
\begin{eqnarray}
    I(T;AB) &=& 2H_T, \nonumber \\
    C_{AB} &=& 1 + H_T,
\end{eqnarray}
where
\begin{equation} \label{eq:target_entropy}
   H_T \coloneqq h(\lambda) = -\lambda\log\lambda - (1-\lambda)\log(1-\lambda)
\end{equation}
is the von Neumann entropy of the reduced state of the target qubit (equivalently, the Shannon entropy of the Schmidt coefficients), and logarithms are taken in base-2 (i.e.,\ entropy and information are measured in bits).
Since the encoding operation is an isometry, the mutual information and superdense capacity between the target and the physical qubits comprising the logical qubit remain unchanged.

\subsection{Superdense Capacity for Stabilizer Subsets}
The first question to consider is: what can the superdense capacity between the target and some subset of the physical qubits be?
It may initially seem like any value from zero to the full $1 + H_T$ bits could be possible, but in fact (as we prove in Appendix~\ref{sec:StabilizerStates}), for any subset $A \subseteq \{1\dots n\}$, the superdense capacity is restricted to three possible values:
\[ C_A \in \{1 - H_T, 1, 1 + H_T\}, \]
corresponding to 
\[ I(T;A) \in \{0, H_T, 2H_T\}. \]

To see that these are the only possible values, consider the prototypical example of superdense coding.
Letting $\lambda = 1/2$, the state~\eqref{eq:encoded_state} becomes
\begin{equation}
    \ket{\psi_\text{enc}} = \frac{1}{\sqrt{2}}\left(\ket{00_L} + \ket{11_L}\right),
\end{equation}
a maximally-entangled Bell state ($H_T = 1$).
As per the standard prescription, the target can send two bits of information by choosing to act on their qubit with one of $\{\Id, \sigma_x, \sigma_y, \sigma_z\}$, thereby transforming the joint state into any of the four Bell states (albeit with one of the qubits encoded as a logical qubit).
When the holder of the encoded qubit receives the target, they can determine which Bell state the system is in (and hence which of the four operations was performed on the target) by measuring the operators $Z\bar{Z}$ and $X\bar{X}$, where $\bar{Z}$ and $\bar{X}$ are logical Paulis acting on the encoded qubit.

If the receiver possesses only a subset $A$ of the physical qubits, one of three cases must hold:
\begin{enumerate}
    \item Both $\bar{Z}$ and $\bar{X}$ are supported on $A$.
    \item Exactly one of $\bar{Z}$ and $\bar{X}$ is supported on $A$.
    \item Neither operator is supported on $A$.
\end{enumerate}
Clearly, in case 1 the receiver can determine exactly which of the four operations was performed on the target and thus gains two bits of information.
This case can be thought of as the ``quantum-advantage'' case: entanglement permits a super-classical amount of information transmission.
In case 2, the receiver can only narrow the possibilities down to two out of four, gaining one bit of information. 
Note that this case is purely classical: when the target and the encoded qubit are unentangled ($H_T \to 0$), all three cases reduce to this one.
Finally, case 3 is the ``quantum-disadvantage'' case: the entanglement between the target and the qubits not in $A$ degrades the ability to transmit even the classically allowed amount of information.

Note that, because of the relation $Y = iXZ$, only two of the logical Paulis are independent.
In the case that two are supported on a given subset, the third is automatically as well.
The focus on $\bar{Z}$ and $\bar{X}$ previously and hereafter is merely a convention.

\subsection{Stabilizer C-PID}
The restricted set of possible values for the superdense capacity of a subset of encoding qubits is a result of the rich structure of stabilizer codes.
One example of a consequence of this structure is the so-called ``cleaning lemma''~\cite{Bravyi_2009,haah2013latticequantumcodesexotic}, which effectively states that any subset of physical qubits either supports both logical operators (while its complement supports none), or both the subset and its complement support the same single operator.
This then has substantial implications for the decomposition into unique, redundant, and synergistic capacities with respect to disjoint subsets $A$ and $B$ (possessed by Alice and Bob).
In fact, there are only eleven possible distinct combinations of capacities $(C_{AB}, C_A, C_B)$, each with a corresponding unique decomposition.
These cases are summarized in Table~\ref{tab:QEC_CPID_all_cases}.
\begin{table}
    \caption{
        \label{tab:QEC_CPID_all_cases}
        The eleven distinct superdense capacity combinations for all $k=1$ stabilizer codes. 
        Every C-PID component is one of the five values $\{0,\ 1-H_T,\ H_T,\ 1,\ 2H_T\}$.
    }
    \begin{ruledtabular}
        \begin{tabular}{ccc|cccc}
            $C_{AB}$ & $C_A$ & $C_B$ & $C_{\text{unq}A}$ & $C_{\text{unq}B}$ & $C_{\text{red}}$ & $C_{\text{syn}}$ \\
            \hline
            $1-H_T$ & $1-H_T$ & $1-H_T$ & $0$    & $0$    & $1-H_T$ & $0$    \\
            $1$     & $1-H_T$ & $1-H_T$ & $0$    & $0$    & $1-H_T$ & $H_T$  \\
            $1$     & $1-H_T$ & $1$     & $0$    & $H_T$  & $1-H_T$ & $0$    \\
            $1$     & $1$     & $1-H_T$ & $H_T$  & $0$    & $1-H_T$ & $0$    \\
            $1$     & $1$     & $1$     & $0$    & $0$    & $1$     & $0$    \\
            $1+H_T$ & $1-H_T$ & $1-H_T$ & $0$    & $0$    & $1-H_T$ & $2H_T$ \\
            $1+H_T$ & $1-H_T$ & $1$     & $0$    & $H_T$  & $1-H_T$ & $H_T$  \\
            $1+H_T$ & $1-H_T$ & $1+H_T$ & $0$    & $2H_T$ & $1-H_T$ & $0$    \\
            $1+H_T$ & $1$     & $1-H_T$ & $H_T$  & $0$    & $1-H_T$ & $H_T$  \\
            $1+H_T$ & $1$     & $1$     & $0$    & $0$    & $1$     & $H_T$  \\
            $1+H_T$ & $1+H_T$ & $1-H_T$ & $2H_T$ & $0$    & $1-H_T$ & $0$    \\
        \end{tabular}
    \end{ruledtabular}
\end{table}

\subsubsection{Unique Information}
From the three cases discussed above, it is clear that one subset possessing unique capacity with respect to another corresponds to that subset having access to a logical operator that the other does not. 
This implies that a subset has no unique information with respect to another if and only if it fails to support more logical operators than the other.
There is another (equivalent) interpretation for a subset possessing unique capacity in terms of erasure correctability: a correctable subset can possess no unique information.
Were a correctable subset to have unique information, then the state could not be perfectly reconstructed after the erasure of those qubits.
A code with distance $d$ can correct up to $d-1$ erasures, therefore no subset of $d-1$ or fewer qubits can possess unique information.
Similarly, a subset whose complement is correctable contains the full $1+H_T$ bits of capacity, thus any subset of $n-(d-1)$ or more qubits contains $2H_T$ bits of unique information (i.e.,\ the full $1+H_T$ bits minus the redundant $1-H_T$ bits).

\subsubsection{Redundant Information}
For general states ($\lambda \text{ not necessarily } 1/2$), the minimum value of the superdense capacity for any subset is $1-H_T$ bits (the ``quantum-disadvantage'' case).
Since any two subsets must have at least this capacity, it is clearly redundant. 
In fact, by the cleaning lemma (see Appendix~\ref{sec:StabilizerStates}), the redundant capacity for any two subsets is either $1 - H_T$ bits, or exactly one bit.
This is of course consistent with the aforementioned exclusion principle for superdense coding---no two components can simultaneously have superdense capacities of greater than 1 bit.

\subsubsection{Synergistic Information}
The presence of synergistic capacity between two subsets has an immediate and obvious interpretation in light of the three cases: two subsets possess synergistic capacity with respect to one another exactly when their union supports a logical operator that neither does individually. 
In partitions for which both $\abs{A},\abs{B} \lesssim n/2$ but $\abs{A}+\abs{B} \gtrsim n - (d-1)$ it is common for neither the $A$ nor the $B$ qubits to have access to any logical operators ($C_A = C_B = 1-H_T$) while their union supports all logicals ($C_{AB} = 1+H_T$), thus exemplifying the interpretation of synergistic information corresponding to ``the whole being greater than the sum of its parts'' (given that $H_T > 1/3$, at least).

\subsection{Example Codes}
\subsubsection{Bit-Flip / Phase-Flip} \label{sssec:Repetition_Codes}
The bit-flip and phase-flip codes are examples of $n$-repetition codes, i.e.,\ codes of the form
\begin{equation}
    \label{eq:rep_code}
    \ket{i} \to \ket{\psi_i}^{\otimes n},
\end{equation}
for some orthogonal set of states $\{\ket{\psi_i}\}$.
For the bit-flip code $\ket{\psi_i} = \ket{i}$, while for the phase-flip the $\ket{\psi_i}$ are the $\ket{+}$ and $\ket{-}$ states.
Clearly, all codes of the form~\eqref{eq:rep_code} are unitarily equivalent and thus have identical information measures.
Stabilizer generators and logical $X$ and $Z$ representatives for the bit-flip code are listed in Table~\ref{tab:BitFlip_stabilizer}, while superdense capacities and corresponding C-PIDs are listed in Table~\ref{tab:BitFlip_CPID}.

For the $[[3,1,1]]$ bit-flip code, states are encoded as
\begin{equation}
    \alpha\ket{0} + \beta\ket{1} \to \alpha\ket{000} + \beta\ket{111}.
\end{equation}
Note that for $\alpha = \beta = 1/\sqrt{2}$ (i.e.,\ $H_T=1$) this is exactly the GHZ state, but even for arbitrary coefficients (and in fact more generally for any $n$) this encoded state still possesses the property that the reduced state resulting from tracing out any qubit is completely unentangled.
The superdense capacity for any strict subset of qubits is precisely one (redundant) bit---the classical case discussed above.
Only when $A$ and $B$ together possess all physical qubits do they have the full capacity of $1+H_T$ bits (1 redundant bit plus an additional $H_T$ synergistic bits).

In terms of the stabilizers and logical operators of these codes, this informational structure is a direct result of the fact that exactly one logical operator ($\bar{X}$ in the case of the bit-flip code) has no representative that acts trivially on any qubit, while the other has single-qubit representatives for every qubit.

\begin{table}
    \caption{
        \label{tab:BitFlip_stabilizer}
        Stabilizer generators and logical operators for the three-qubit bit-flip code.
    }
    \begin{ruledtabular}
        \begin{tabular}{c|ccc}
            $g_1$ & $Z$ & $Z$ & $I$ \\
            $g_2$ & $I$ & $Z$ & $Z$ \\
            \hline
            $\bar{Z}$ & $Z$ & $I$ & $I$ \\
            $\bar{X}$ & $X$ & $X$ & $X$ \\
        \end{tabular}
    \end{ruledtabular}
\end{table}

\begin{table}
    \caption{
        \label{tab:BitFlip_CPID}
        Superdense capacities and C-PID for the two possible partitions of the bit-flip and phase-flip codes. Any strict subset of the physical qubits has 1 redundant bit of information, and only all qubits collectively carry the full $1+H_T$ bits.
    }
    \begin{ruledtabular}
        \begin{tabular}{c|ccc|cccc}
            $(\abs{A},\abs{B})$ & $C_{AB}$ & $C_A$ & $C_B$ & $C_{\text{unq}A}$ & $C_{\text{unq}B}$ & $C_{\text{red}}$ & $C_{\text{syn}}$ \\
            \hline
            $(1,1)$ & $1$     & $1$ & $1$ & $0$ & $0$ & $1$ & $0$ \\
            \hline
            $(1,2)$ & $1+H_T$ & $1$ & $1$ & $0$ & $0$ & $1$ & $H_T$ \\
        \end{tabular}
    \end{ruledtabular}
\end{table}

\subsubsection{The Five-Qubit Perfect Code}
The five-qubit code is the smallest code capable of correcting arbitrary single-qubit errors. 
It is an example of a ``perfect'' code, as it saturates the quantum Hamming bound~\cite{laflamme1996perfectquantumerrorcorrection}, and is also the only such $[[5,1,3]]$ code up to local Clifford operations and qubit permutations.
Stabilizer generators and logical operators for the five-qubit code are displayed in Table~\ref{tab:FiveQubit_stabilizer}.

Superdense capacities and the corresponding C-PID for all partitions of the five-qubit code are listed in Table~\ref{tab:FiveQubit_CPID}.
Note that there is sufficient symmetry that the capacity of a subset of qubits is determined only by the size of the subset and not the specific qubits it contains.
Specifically, any subset of three or more qubits contains all information about the target ($C=1+H_T$), while any two or fewer qubits constitute a correctable set ($C=1-H_T$).
In fact, this is precisely the well-known $((3,5))$ Quantum Secret Sharing capability of this code~\cite{Cleve_1999}.
Since the code has five qubits, $\abs{B} \geq 3$ implies $\abs{A} \leq 2$, so Bob will have unique information exactly when he holds more than half the qubits.
Synergistic information is present only when both $\abs{A},\abs{B}\leq 2$, but $\abs{A}+\abs{B}>2$.
Note also that the classical case $C=1$ is precluded: any subset is either correctable or its complement is.

\begin{table}
    \caption{
        \label{tab:FiveQubit_stabilizer}
        Stabilizer generators and logical operators for the five-qubit perfect code.
    }
    \begin{ruledtabular}
        \begin{tabular}{c|ccccc}
            $g_1$ & $X$ & $Z$ & $Z$ & $X$ & $I$ \\
            $g_2$ & $I$ & $X$ & $Z$ & $Z$ & $X$ \\
            $g_3$ & $X$ & $I$ & $X$ & $Z$ & $Z$ \\
            $g_4$ & $Z$ & $X$ & $I$ & $X$ & $Z$ \\
            \hline
            $\bar{Z}$ & $Z$ & $Z$ & $Z$ & $Z$ & $Z$ \\
            $\bar{X}$ & $X$ & $X$ & $X$ & $X$ & $X$ \\
        \end{tabular}
    \end{ruledtabular}
\end{table}
\begin{table*}[h]
    \caption{
        \label{tab:FiveQubit_CPID}
        C-PID for all possible groupings of the physical qubits for the five-qubit code.
        Note that the five-qubit code is sufficiently symmetric that the superdense capacities and corresponding PIDs depend only on the sizes of the A/B subsets, and not \textit{which} qubits are in them.
    }
    \begin{ruledtabular}
        \begin{tabular}{c|ccc|cccc}
            $(\abs{A},\abs{B})$ & $C_{AB}$ & $C_A$ & $C_B$ & $C_{\text{unq}A}$ & $C_{\text{unq}B}$ & $C_{\text{red}}$ & $C_{\text{syn}}$ \\
            \hline
            $(1,1)$ & $1-H_T$ & $1-H_T$ & $1-H_T$ & $0$ & $0$    & $1-H_T$ & $0$ \\
            $(1,2)$ & $1+H_T$ & $1-H_T$ & $1-H_T$ & $0$ & $0$    & $1-H_T$ & $2H_T$ \\
            $(1,3)$ & $1+H_T$ & $1-H_T$ & $1+H_T$ & $0$ & $2H_T$ & $1-H_T$ & $0$ \\
            $(1,4)$ & $1+H_T$ & $1-H_T$ & $1+H_T$ & $0$ & $2H_T$ & $1-H_T$ & $0$ \\
            $(2,2)$ & $1+H_T$ & $1-H_T$ & $1-H_T$ & $0$ & $0$    & $1-H_T$ & $2H_T$ \\
            $(2,3)$ & $1+H_T$ & $1-H_T$ & $1+H_T$ & $0$ & $2H_T$ & $1-H_T$ & $0$ \\
        \end{tabular}
    \end{ruledtabular}
\end{table*}

\subsubsection{The Steane Code}
The $[[7,1,3]]$ Steane code is an example of a self-dual CSS code~\cite{calderbank1996, steane1996multiple} constructed from a classical $[7,4,3]$ Hamming code.
Stabilizer generators and logical operators for the Steane code are listed in Table~\ref{tab:Steane_stabilizer}, while superdense capacities and their corresponding C-PIDs for representative partitions are listed in Table~\ref{tab:Steane_CPID}.

Unlike the repetition codes and the five-qubit code, the superdense capacities for a subset of the Steane code's physical qubits depend not only on its size, but also on which qubits it contains.
For instance, the 3-qubit subset $\{1,2,3\}$ supports all logical operators, while the subset $\{1,2,4\}$ supports none.

As a distance-3 code, any subset of 2 or fewer qubits is correctable ($C=1-H_T$) while any subset of 5 or more qubits possesses the full $1+H_T$ bits of capacity.
Subsets of size 3-4 could then in principle exhibit the classical case ($C=1$), but this is in fact ruled out due to the code being self-dual: the logical $X$ and $Z$ operators have identical support---no subset can support one but not the other.

\begin{table*}[h]
    \caption{
        \label{tab:Steane_stabilizer}
        Stabilizer generators and logical operators for the seven-qubit Steane code.
    }
    \begin{ruledtabular}
        \begin{tabular}{c|ccccccc}
            $g_1$ & $I$ & $I$ & $I$ & $X$ & $X$ & $X$ & $X$ \\
            $g_2$ & $I$ & $X$ & $X$ & $I$ & $I$ & $X$ & $X$ \\
            $g_3$ & $X$ & $I$ & $X$ & $I$ & $X$ & $I$ & $X$ \\
            $g_4$ & $I$ & $I$ & $I$ & $Z$ & $Z$ & $Z$ & $Z$ \\
            $g_5$ & $I$ & $Z$ & $Z$ & $I$ & $I$ & $Z$ & $Z$ \\
            $g_6$ & $Z$ & $I$ & $Z$ & $I$ & $Z$ & $I$ & $Z$ \\
            \hline
            $\bar{Z}$ & $Z$ & $Z$ & $Z$ & $Z$ & $Z$ & $Z$ & $Z$ \\
            $\bar{X}$ & $X$ & $X$ & $X$ & $X$ & $X$ & $X$ & $X$ \\
        \end{tabular}
    \end{ruledtabular}
\end{table*}
\begin{table*}[h]
    \caption{
        \label{tab:Steane_CPID}
        Example partitions and corresponding superdense capacities and PIDs for the seven-qubit Steane code.
        Note that, unlike the five-qubit code, the subset sizes alone do not determine the capacities and decompositions. For example, there are two kinds of partitions in which $\abs{A}=1$ and $\abs{B}=2$.
    }
    \begin{ruledtabular}
        \begin{tabular}{c|cc|ccc|cccc}
            $(\abs{A},\abs{B})$ & $A$ & $B$ & $C_{AB}$ & $C_A$ & $C_B$ & $C_{\text{unq}A}$ & $C_{\text{unq}B}$ & $C_{\text{red}}$ & $C_{\text{syn}}$ \\
            \hline
            $(1,1)$ & $\{1\}$ & $\{2\}$             & $1-H_T$ & $1-H_T$ & $1-H_T$ & $0$    & $0$    & $1-H_T$ & $0$ \\
            \hline
            $(1,2)$ & $\{1\}$ & $\{2,3\}$           & $1+H_T$ & $1-H_T$ & $1-H_T$ & $0$    & $0$    & $1-H_T$ & $2H_T$ \\
                    & $\{1\}$ & $\{2,4\}$           & $1-H_T$ & $1-H_T$ & $1-H_T$ & $0$    & $0$    & $1-H_T$ & $0$ \\
            \hline
            $(1,3)$ & $\{1\}$ & $\{2,3,4\}$         & $1+H_T$ & $1-H_T$ & $1-H_T$ & $0$    & $0$    & $1-H_T$ & $2H_T$ \\
                    & $\{1\}$ & $\{2,4,6\}$         & $1+H_T$ & $1-H_T$ & $1+H_T$ & $0$    & $2H_T$ & $1-H_T$ & $0$ \\
                    & $\{1\}$ & $\{2,4,7\}$         & $1-H_T$ & $1-H_T$ & $1-H_T$ & $0$    & $0$    & $1-H_T$ & $0$ \\
            \hline
            $(1,4)$ & $\{1\}$ & $\{2,3,4,5\}$       & $1+H_T$ & $1-H_T$ & $1-H_T$ & $0$    & $0$    & $1-H_T$ & $2H_T$ \\
                    & $\{1\}$ & $\{2,3,4,6\}$       & $1+H_T$ & $1-H_T$ & $1+H_T$ & $0$    & $2H_T$ & $1-H_T$ & $0$ \\
            \hline
            $(1,5)$ & $\{1\}$ & $\{2,3,4,5,6\}$     & $1+H_T$ & $1-H_T$ & $1+H_T$ & $0$    & $2H_T$ & $1-H_T$ & $0$ \\
            \hline
            $(1,6)$ & $\{1\}$ & $\{2,3,4,5,6,7\}$   & $1+H_T$ & $1-H_T$ & $1+H_T$ & $0$    & $2H_T$ & $1-H_T$ & $0$ \\
            \hline
            $(2,2)$ & $\{1,2\}$ & $\{3,4\}$         & $1+H_T$ & $1-H_T$ & $1-H_T$ & $0$    & $0$    & $1-H_T$ & $2H_T$ \\
                    & $\{1,2\}$ & $\{4,7\}$         & $1-H_T$ & $1-H_T$ & $1-H_T$ & $0$    & $0$    & $1-H_T$ & $0$ \\
            \hline
            $(2,3)$ & $\{1,2\}$ & $\{3,4,5\}$       & $1+H_T$ & $1-H_T$ & $1-H_T$ & $0$    & $0$    & $1-H_T$ & $2H_T$ \\
                    & $\{1,2\}$ & $\{3,4,7\}$       & $1+H_T$ & $1-H_T$ & $1+H_T$ & $0$    & $2H_T$ & $1-H_T$ & $0$ \\
            \hline
            $(2,4)$ & $\{1,2\}$ & $\{3,4,5,6\}$     & $1+H_T$ & $1-H_T$ & $1+H_T$ & $0$    & $2H_T$ & $1-H_T$ & $0$ \\
                    & $\{1,2\}$ & $\{4,5,6,7\}$     & $1+H_T$ & $1-H_T$ & $1-H_T$ & $0$    & $0$    & $1-H_T$ & $2H_T$ \\
            \hline
            $(2,5)$ & $\{1,2\}$ & $\{3,4,5,6,7\}$   & $1+H_T$ & $1-H_T$ & $1+H_T$ & $0$    & $2H_T$ & $1-H_T$ & $0$ \\
            \hline
            $(3,3)$ & $\{1,2,3\}$ & $\{4,5,6\}$     & $1+H_T$ & $1+H_T$ & $1-H_T$ & $2H_T$ & $0$    & $1-H_T$ & $0$ \\
                    & $\{1,2,4\}$ & $\{3,5,6\}$     & $1+H_T$ & $1-H_T$ & $1+H_T$ & $0$    & $2H_T$ & $1-H_T$ & $0$ \\
                    & $\{1,2,4\}$ & $\{3,5,7\}$     & $1+H_T$ & $1-H_T$ & $1-H_T$ & $0$    & $0$    & $1-H_T$ & $2H_T$ \\
            \hline
            $(3,4)$ & $\{1,2,3\}$ & $\{4,5,6,7\}$   & $1+H_T$ & $1+H_T$ & $1-H_T$ & $2H_T$ & $0$    & $1-H_T$ & $0$ \\
                    & $\{1,2,4\}$ & $\{3,5,6,7\}$   & $1+H_T$ & $1-H_T$ & $1+H_T$ & $0$    & $2H_T$ & $1-H_T$ & $0$ \\
        \end{tabular}
    \end{ruledtabular}
\end{table*}

\subsubsection{The Shor Code}
The Shor code is a $[[9,1,3]]$ code formed by concatenating the bit-flip and phase-flip codes.
Specifically, the logical qubit is first phase-flip encoded into three qubits, each of which is then bit-flip encoded into a block of three.
Stabilizer generators and logical operators for the Shor code are listed in Table~\ref{tab:Shor_stabilizer}; superdense capacities and corresponding C-PIDs are listed in Table~\ref{tab:Shor_CPID}.

The Shor code has a sufficiently rich structure that all eleven possible superdense capacity combinations are realized for various partitions of the physical qubits. 
The concatenated structure of the code is apparent in the various $A$/$B$ partitions.
Specifically, when a subset aligns with the block structure of either the inner or outer code, its capacity matches that of the corresponding subset in the unconcatenated code.
The subset $\{1,2,3\}$, for instance, corresponds to the first block of the inner (bit-flip) code and possesses one bit of capacity just as all individual qubits of the bit-flip code do.
Alternatively, a subset containing at least one qubit from each of $\{\{1,2,3\}, \{4,5,6\}, \{7,8,9\}\}$ corresponds to a qubit of the outer (phase-flip) code, and hence also possesses a capacity of one bit.

\begin{table}[h]
    \caption{
        \label{tab:Shor_stabilizer}
        Stabilizer generators and logical operators for the nine-qubit Shor code.
    }
    \begin{ruledtabular}
        \begin{tabular}{c|ccccccccc}
            $g_1$ & $Z$ & $Z$ & $I$ & $I$ & $I$ & $I$ & $I$ & $I$ & $I$ \\
            $g_2$ & $I$ & $Z$ & $Z$ & $I$ & $I$ & $I$ & $I$ & $I$ & $I$ \\
            $g_3$ & $I$ & $I$ & $I$ & $Z$ & $Z$ & $I$ & $I$ & $I$ & $I$ \\
            $g_4$ & $I$ & $I$ & $I$ & $I$ & $Z$ & $Z$ & $I$ & $I$ & $I$ \\
            $g_5$ & $I$ & $I$ & $I$ & $I$ & $I$ & $I$ & $Z$ & $Z$ & $I$ \\
            $g_6$ & $I$ & $I$ & $I$ & $I$ & $I$ & $I$ & $I$ & $Z$ & $Z$ \\
            $g_7$ & $X$ & $X$ & $X$ & $X$ & $X$ & $X$ & $I$ & $I$ & $I$ \\
            $g_8$ & $I$ & $I$ & $I$ & $X$ & $X$ & $X$ & $X$ & $X$ & $X$ \\
            \hline
            $\bar{Z}$ & $X$ & $X$ & $X$ & $X$ & $X$ & $X$ & $X$ & $X$ & $X$ \\
            $\bar{X}$ & $Z$ & $Z$ & $Z$ & $Z$ & $Z$ & $Z$ & $Z$ & $Z$ & $Z$ \\
        \end{tabular}
    \end{ruledtabular}
\end{table}
\begin{table*}[h]
    \caption{
        \label{tab:Shor_CPID}
        C-PID for the nine-qubit Shor code for all inequivalent partitions, labeled by the specific qubit subsets $A$ and $B$.
    }
    \begin{ruledtabular}
        \begin{tabular}{c|cc|ccc|cccc}
            $(\abs{A},\abs{B})$ & $A$ & $B$ & $C_{AB}$ & $C_A$ & $C_B$ & $C_{\text{unq}A}$ & $C_{\text{unq}B}$ & $C_{\text{red}}$ & $C_{\text{syn}}$ \\
            \hline
            $(1,1)$ & $\{1\}$ & $\{2\}$                 & $1-H_T$ & $1-H_T$ & $1-H_T$ & $0$    & $0$    & $1-H_T$ & $0$ \\
            \hline
            $(1,2)$ & $\{1\}$ & $\{2,3\}$               & $1$     & $1-H_T$ & $1-H_T$ & $0$    & $0$    & $1-H_T$ & $H_T$ \\
                    & $\{1\}$ & $\{2,4\}$               & $1-H_T$ & $1-H_T$ & $1-H_T$ & $0$    & $0$    & $1-H_T$ & $0$ \\
            \hline
            $(1,3)$ & $\{1\}$ & $\{2,3,4\}$             & $1$     & $1-H_T$ & $1-H_T$ & $0$    & $0$    & $1-H_T$ & $H_T$ \\
                    & $\{1\}$ & $\{2,4,5\}$             & $1-H_T$ & $1-H_T$ & $1-H_T$ & $0$    & $0$    & $1-H_T$ & $0$ \\
                    & $\{1\}$ & $\{2,4,7\}$             & $1$     & $1-H_T$ & $1$     & $0$    & $H_T$  & $1-H_T$ & $0$ \\
            \hline
            $(1,4)$ & $\{1\}$ & $\{2,3,4,5\}$           & $1$     & $1-H_T$ & $1-H_T$ & $0$    & $0$    & $1-H_T$ & $H_T$ \\
                    & $\{1\}$ & $\{2,3,4,7\}$           & $1+H_T$ & $1-H_T$ & $1$     & $0$    & $H_T$  & $1-H_T$ & $H_T$ \\
                    & $\{1\}$ & $\{2,4,5,6\}$           & $1$     & $1-H_T$ & $1$     & $0$    & $H_T$  & $1-H_T$ & $0$ \\
            \hline
            $(1,5)$ & $\{1\}$ & $\{2,3,4,5,6\}$         & $1$     & $1-H_T$ & $1$     & $0$    & $H_T$  & $1-H_T$ & $0$ \\
                    & $\{1\}$ & $\{2,3,4,5,7\}$         & $1+H_T$ & $1-H_T$ & $1$     & $0$    & $H_T$  & $1-H_T$ & $H_T$ \\
                    & $\{1\}$ & $\{2,4,5,6,7\}$         & $1+H_T$ & $1-H_T$ & $1+H_T$ & $0$    & $2H_T$ & $1-H_T$ & $0$ \\
            \hline
            $(1,6)$ & $\{1\}$ & $\{2,3,4,5,6,7\}$       & $1+H_T$ & $1-H_T$ & $1+H_T$ & $0$    & $2H_T$ & $1-H_T$ & $0$ \\
                    & $\{1\}$ & $\{2,3,4,5,7,8\}$       & $1+H_T$ & $1-H_T$ & $1$     & $0$    & $H_T$  & $1-H_T$ & $H_T$ \\
            \hline
            $(1,7)$ & $\{1\}$ & $\{2,3,4,5,6,7,8\}$     & $1+H_T$ & $1-H_T$ & $1+H_T$ & $0$    & $2H_T$ & $1-H_T$ & $0$ \\
            \hline
            $(1,8)$ & $\{1\}$ & $\{2,3,4,5,6,7,8,9\}$   & $1+H_T$ & $1-H_T$ & $1+H_T$ & $0$    & $2H_T$ & $1-H_T$ & $0$ \\
            \hline
            $(2,2)$ & $\{1,2\}$ & $\{3,4\}$             & $1$     & $1-H_T$ & $1-H_T$ & $0$    & $0$    & $1-H_T$ & $H_T$ \\
                    & $\{1,2\}$ & $\{4,5\}$             & $1-H_T$ & $1-H_T$ & $1-H_T$ & $0$    & $0$    & $1-H_T$ & $0$ \\
            \hline
            $(2,3)$ & $\{1,2\}$ & $\{3,4,5\}$           & $1$     & $1-H_T$ & $1-H_T$ & $0$    & $0$    & $1-H_T$ & $H_T$ \\
                    & $\{1,2\}$ & $\{3,4,7\}$           & $1+H_T$ & $1-H_T$ & $1$     & $0$    & $H_T$  & $1-H_T$ & $H_T$ \\
                    & $\{1,2\}$ & $\{4,5,6\}$           & $1$     & $1-H_T$ & $1$     & $0$    & $H_T$  & $1-H_T$ & $0$ \\
                    & $\{1,4\}$ & $\{2,3,7\}$           & $1+H_T$ & $1-H_T$ & $1-H_T$ & $0$    & $0$    & $1-H_T$ & $2H_T$ \\
            \hline
            $(2,4)$ & $\{1,2\}$ & $\{3,4,5,6\}$         & $1$     & $1-H_T$ & $1$     & $0$    & $H_T$  & $1-H_T$ & $0$ \\
                    & $\{1,2\}$ & $\{3,4,5,7\}$         & $1+H_T$ & $1-H_T$ & $1$     & $0$    & $H_T$  & $1-H_T$ & $H_T$ \\
                    & $\{1,2\}$ & $\{4,5,7,8\}$         & $1$     & $1-H_T$ & $1-H_T$ & $0$    & $0$    & $1-H_T$ & $H_T$ \\
                    & $\{1,4\}$ & $\{2,3,7,8\}$         & $1+H_T$ & $1-H_T$ & $1-H_T$ & $0$    & $0$    & $1-H_T$ & $2H_T$ \\
            \hline
            $(2,5)$ & $\{1,2\}$ & $\{3,4,5,6,7\}$       & $1+H_T$ & $1-H_T$ & $1+H_T$ & $0$    & $2H_T$ & $1-H_T$ & $0$ \\
                    & $\{1,2\}$ & $\{3,4,5,7,8\}$       & $1+H_T$ & $1-H_T$ & $1$     & $0$    & $H_T$  & $1-H_T$ & $H_T$ \\
            \hline
            $(2,6)$ & $\{1,2\}$ & $\{3,4,5,6,7,8\}$     & $1+H_T$ & $1-H_T$ & $1+H_T$ & $0$    & $2H_T$ & $1-H_T$ & $0$ \\
                    & $\{1,2\}$ & $\{4,5,6,7,8,9\}$     & $1+H_T$ & $1-H_T$ & $1$     & $0$    & $H_T$  & $1-H_T$ & $H_T$ \\
            \hline
            $(2,7)$ & $\{1,2\}$ & $\{3,4,5,6,7,8,9\}$   & $1+H_T$ & $1-H_T$ & $1+H_T$ & $0$    & $2H_T$ & $1-H_T$ & $0$ \\
            \hline
            $(3,3)$ & $\{1,2,3\}$ & $\{4,5,6\}$         & $1$     & $1$     & $1$     & $0$    & $0$    & $1$     & $0$ \\
                    & $\{1,2,3\}$ & $\{4,5,7\}$         & $1+H_T$ & $1$     & $1-H_T$ & $H_T$  & $0$    & $1-H_T$ & $H_T$ \\
                    & $\{1,2,4\}$ & $\{3,5,6\}$         & $1$     & $1-H_T$ & $1-H_T$ & $0$    & $0$    & $1-H_T$ & $H_T$ \\
                    & $\{1,2,4\}$ & $\{3,5,7\}$         & $1+H_T$ & $1-H_T$ & $1$     & $0$    & $H_T$  & $1-H_T$ & $H_T$ \\
                    & $\{1,2,4\}$ & $\{3,7,8\}$         & $1+H_T$ & $1-H_T$ & $1-H_T$ & $0$    & $0$    & $1-H_T$ & $2H_T$ \\
            \hline
            $(3,4)$ & $\{1,2,3\}$ & $\{4,5,6,7\}$       & $1+H_T$ & $1$     & $1$     & $0$    & $0$    & $1$     & $H_T$ \\
                    & $\{1,2,3\}$ & $\{4,5,7,8\}$       & $1+H_T$ & $1$     & $1-H_T$ & $H_T$  & $0$    & $1-H_T$ & $H_T$ \\
                    & $\{1,2,4\}$ & $\{3,5,6,7\}$       & $1+H_T$ & $1-H_T$ & $1$     & $0$    & $H_T$  & $1-H_T$ & $H_T$ \\
                    & $\{1,2,4\}$ & $\{5,6,7,8\}$       & $1+H_T$ & $1-H_T$ & $1-H_T$ & $0$    & $0$    & $1-H_T$ & $2H_T$ \\
            \hline
            $(3,5)$ & $\{1,2,3\}$ & $\{4,5,6,7,8\}$     & $1+H_T$ & $1$     & $1$     & $0$    & $0$    & $1$     & $H_T$ \\
                    & $\{1,2,4\}$ & $\{3,5,6,7,8\}$     & $1+H_T$ & $1-H_T$ & $1$     & $0$    & $H_T$  & $1-H_T$ & $H_T$ \\
                    & $\{1,2,4\}$ & $\{3,5,7,8,9\}$     & $1+H_T$ & $1-H_T$ & $1+H_T$ & $0$    & $2H_T$ & $1-H_T$ & $0$ \\
            \hline
            $(3,6)$ & $\{1,2,3\}$ & $\{4,5,6,7,8,9\}$   & $1+H_T$ & $1$     & $1$     & $0$    & $0$    & $1$     & $H_T$ \\
                    & $\{1,2,4\}$ & $\{3,5,6,7,8,9\}$   & $1+H_T$ & $1-H_T$ & $1+H_T$ & $0$    & $2H_T$ & $1-H_T$ & $0$ \\
            \hline
            $(4,4)$ & $\{1,2,3,4\}$ & $\{5,6,7,8\}$     & $1+H_T$ & $1$     & $1-H_T$ & $H_T$  & $0$    & $1-H_T$ & $H_T$ \\
                    & $\{1,2,3,4\}$ & $\{5,7,8,9\}$     & $1+H_T$ & $1$     & $1$     & $0$    & $0$    & $1$     & $H_T$ \\
                    & $\{1,2,4,5\}$ & $\{3,6,7,8\}$     & $1+H_T$ & $1-H_T$ & $1$     & $0$    & $H_T$  & $1-H_T$ & $H_T$ \\
            \hline
            $(4,5)$ & $\{1,2,3,4\}$ & $\{5,6,7,8,9\}$   & $1+H_T$ & $1$     & $1$     & $0$    & $0$    & $1$     & $H_T$ \\
                    & $\{1,2,4,5\}$ & $\{3,6,7,8,9\}$   & $1+H_T$ & $1-H_T$ & $1+H_T$ & $0$    & $2H_T$ & $1-H_T$ & $0$ \\
        \end{tabular}
    \end{ruledtabular}
\end{table*}

\subsection{Comparison to Random Codes}
Random QEC codes consist of Haar-random isometries mapping the logical qubit to $n$ physical qubits.
When compared to random codes mapping to the same number of physical qubits, exact codes are ``extremal'' in the sense that they exhibit either maximal or minimal values of unique, redundant, and synergistic capacities.
In particular, for any qubit partition random codes always have more redundant information than exact codes.
Additionally, while exact codes can never have two disjoint subsets both with unique information, similarly sized disjoint subsets for random codes tend to both have small amounts of unique information.
Distributions of the C-PID quantities for Haar-random five-qubit codes are depicted in Figure~\ref{fig:RandCodeDists5}, with the corresponding values for the exact five-qubit code plotted for comparison.

\begin{figure*}
    \includegraphics[scale=0.6]{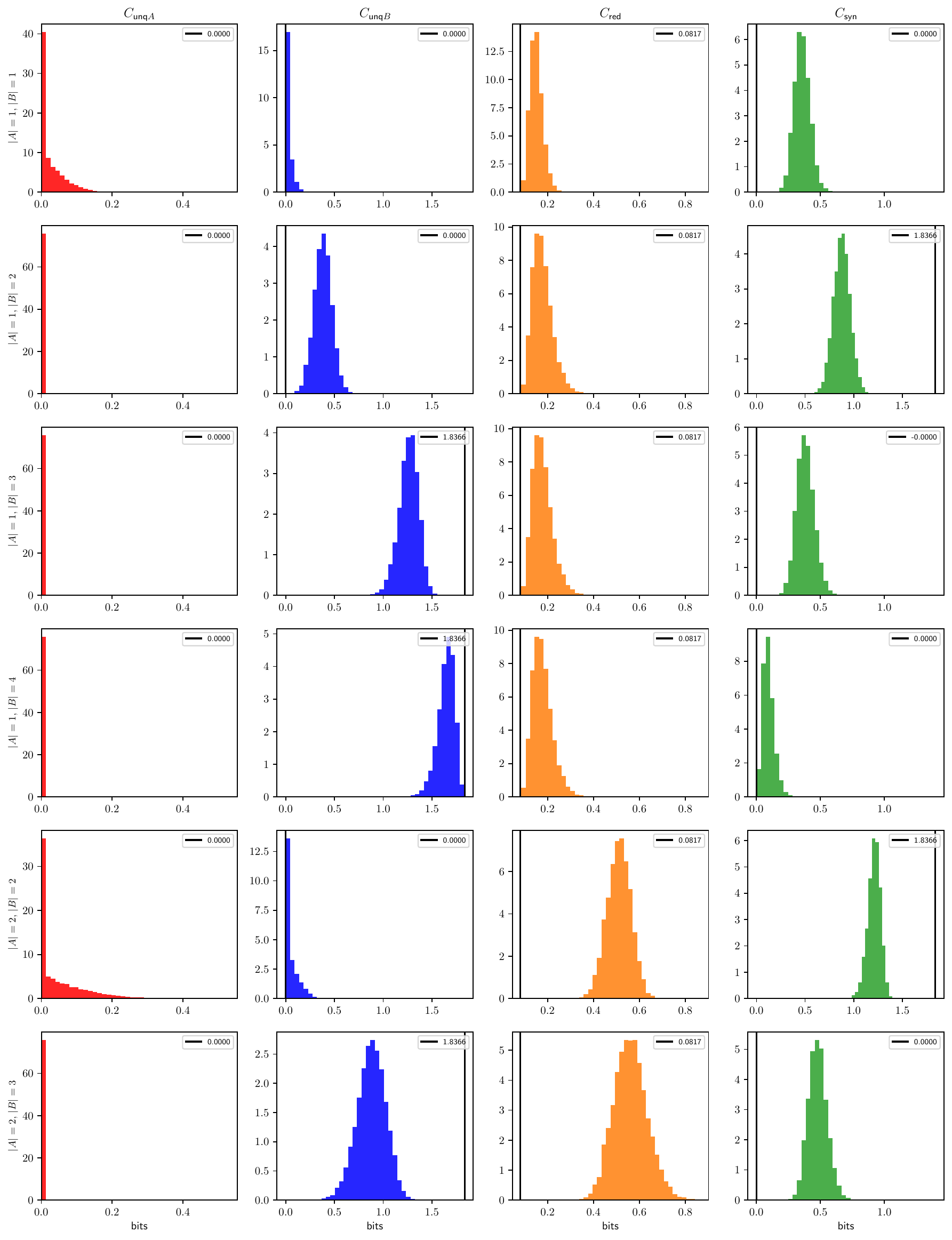}
    \caption{
        \label{fig:RandCodeDists5}
        Distributions of $C_{\text{unq}}$, $C_{\text{red}}$, and $C_{\text{syn}}$ for 10000 random $[[5,1]]$ encodings for the state given in~\eqref{eq:psi_0} with $\lambda = 1/3$.
        The corresponding values of the four C-PID quantities for the same partitions of the exact five-qubit code are displayed as vertical lines.
    }
\end{figure*}

Subsets of sufficiently small size have little or no unique information, particularly with respect to larger subsets.
This is consistent with the observation in~\cite{ma2025haarrandomcodesattain} that random codes approximate ``good'' codes well.
Note that subsets tend to only achieve a significant amount of unique information when the number of qubits is on the order of $n/2$ or greater.
This is in fact consistent with the results of the Hayden and Preskill thought experiment~\cite{Hayden_2007} in which an arbitrary quantum state thrown into a black hole can be recovered (with high fidelity) once roughly half of the Hawking-radiated qubits have been collected.

\section{\label{sec:Decoherence} Application: Decoherence Model}
In this Section we apply our PID to the type of decoherence model that has been used~\cite{Touil_2022} to investigate under what conditions quantum-mechanical systems display classical features.
In this context, one specific physical model that Riedel and Zurek~\cite{riedel2010} analyzed is that of photons interacting (under everyday circumstances) with some quantum system, with the goal of estimating decoherence rates and how many times the same information is redundantly encoded in the photons. 
Here we extend that model to two different types of particles interacting with the system, say photons and nitrogen molecules in the air, to mention one specific everyday example. 
When the two types of particles interact in different ways with the system one really does need a definition for redundant information in order to be able to distinguish it from non-redundant information. 
We will first discuss our definition in the case of just one type of interaction to see whether our results agree with previous results, and then move on to discuss our findings for the case of two different interactions.

\subsection{The Model}
In the decoherence model, the system $\mathcal{S}$ is represented by a qubit initially in the state $\ket{\psi_\mathcal{S}} = \sqrt{\lambda}\ket{0} + \sqrt{1-\lambda}\ket{1}$, while the environment $\mathcal{E}$ is modeled by $N \gg 1$ non-interacting qubits initially in the state $\ket{0}$.
The environment qubits then interact with the system through an imperfect C-NOT gate,
\begin{equation} \label{eq:c-maybe}
    CX_\alpha = \dyad{0}\otimes\mathbbm{1} + \dyad{1}\otimes X_\alpha,
\end{equation}
with
\begin{equation} \label{eq:x-alpha}
    X_\alpha = \mqty(\sin\alpha & \cos\alpha \\ \cos\alpha & -\sin\alpha)
\end{equation}
where $\alpha$ parametrizes the imperfection of the gate.

After interacting, the joint state of the system and environment is the entangled state
\begin{equation}
    \ket{\psi_{\mathcal{SE}}} = \sqrt{\lambda}\ket{0}\ket{0}^{\otimes N} + \sqrt{1-\lambda}\ket{1}\ket{\tilde{1}}^{\otimes N},
\end{equation}
where
\begin{equation}
    \ket{\tilde{1}} = \sin\alpha\ket{0} + \cos\alpha\ket{1}.
\end{equation}
A key feature of this model is the emergence of what the authors refer to as the ``classical plateau'' (Figure~\ref{fig:ZurekPlateau}) of the mutual information between the system and an environment fragment as a function of fragment size $m$.
The level of this plateau is determined by the entropy of the system and its presence indicates that classical information about the system is completely redundantly encoded into the environment.

\begin{figure}
    \includegraphics[width=\columnwidth]{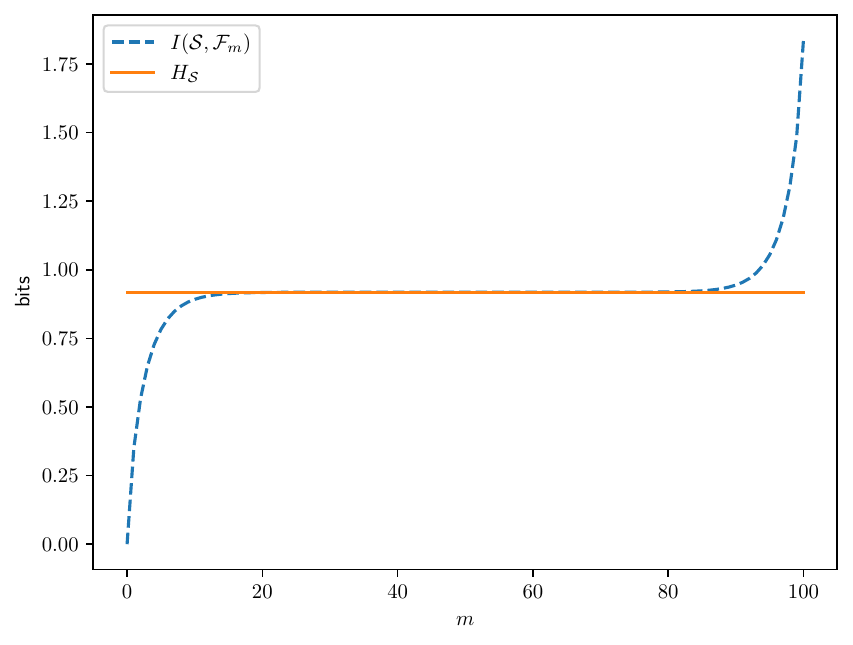}
    \caption{
        \label{fig:ZurekPlateau} 
        Mutual information between system and environment fragment as a function of fragment size with $\lambda=1/3$, $\sin\alpha = 0.85$, and $N = 100$.
        For moderately sized fragments ($m$ around 15) the fragment/system mutual information saturates at the system entropy.
        The fragment size needed to reach the plateau is determined by the imperfection parameter $\alpha$.
    }
\end{figure}

\subsection{Redundant Information in the Decoherence Model}
To apply the C-PID to the decoherence model, the system qubit $\mathcal{S}$ is identified as the target, and we consider Alice as possessing a fragment $\mathcal{F}_m$ of the environment consisting of $m$ qubits.
Bob then possesses the remainder of the environment qubits.

The collective and individual superdense capacities for the decoherence model are plotted in Figure~\ref{fig:OriginalZurekCapacities}.
Recall that the individual superdense capacity $C_A$ is directly related to the target-Alice mutual information by
\begin{equation}
    C_A = I(T;A) + (1 - H_T),
\end{equation}
so the value of the individual superdense capacity on the plateau is 1 bit, regardless of the target entropy.

The PID for the decoherence model is plotted in Figure~\ref{fig:OriginalZurekCPID}.
\begin{figure*}
    \subfloat[\label{fig:OriginalZurekCapacities}]{%
        \includegraphics[width=\columnwidth]{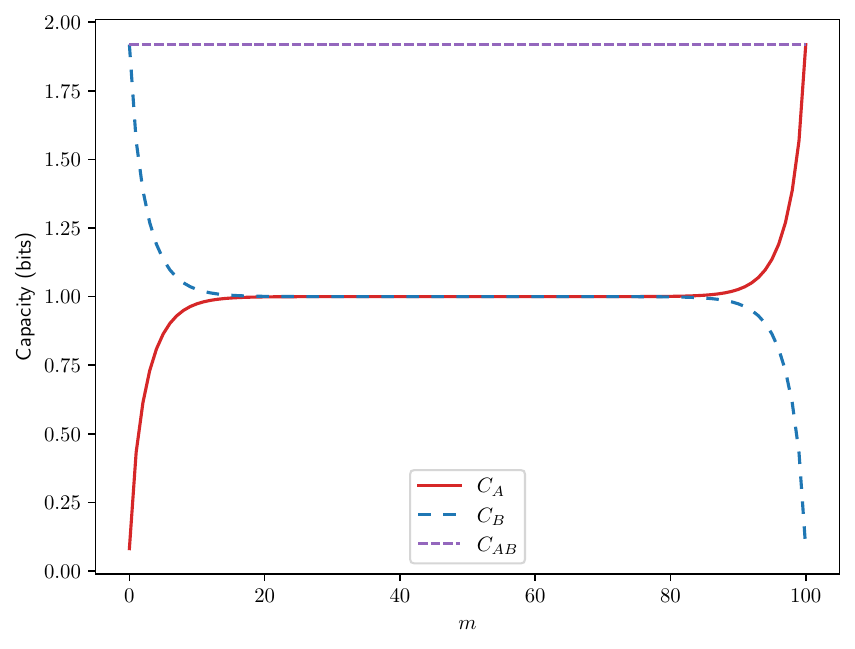}%
    }
    \hspace*{\fill}%
    \subfloat[\label{fig:OriginalZurekCPID}]{%
        \includegraphics[width=\columnwidth]{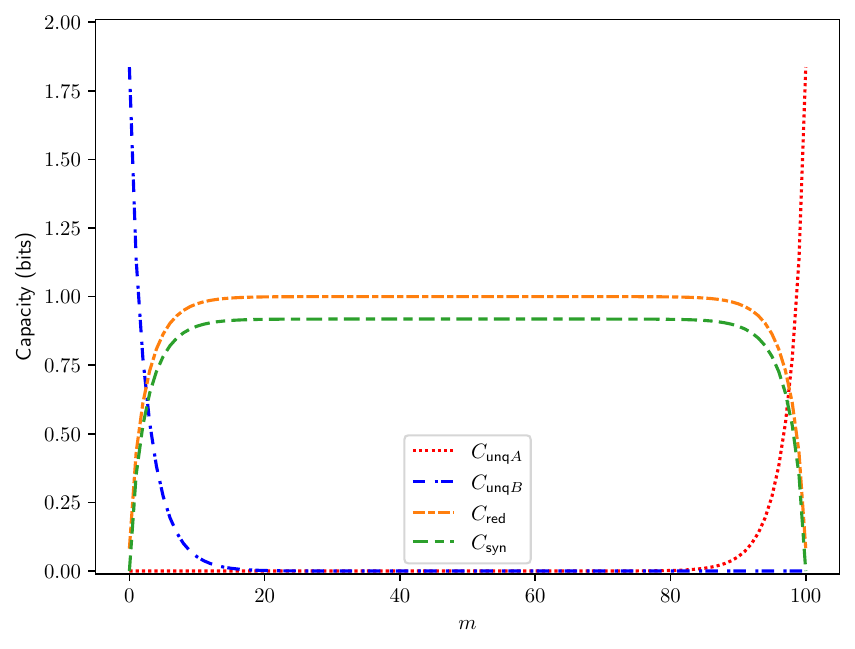}%
    }
    \caption{
        a) Collective and individual superdense capacities for the decoherence model as a function of the number of environment qubits in Alice's possession, with $\lambda=1/3$, $\sin\alpha = 0.85$, and $N=100$ environment qubits total.
        The collective two-bit capacity depends only on the state of the target and not how the environment qubits are partitioned between Alice and Bob. It is maximized for $\lambda = 1/2$.
        As Alice and Bob interact with the target in the same way, their individual two-bit capacities are symmetric under $A \leftrightarrow B$ and $m \leftrightarrow N-m$.   
        b) C-PID for the original decoherence model two-bit capacity.
        When either party possesses the vast majority of the environment qubits, their information is almost entirely unique.
        On the classical plateau, Alice's and Bob's individual capacities of one bit are each entirely redundant.
        Synergistic capacity then accounts for the remainder of the collective two-bit capacity.
    }
\end{figure*}
Notably, the value of the redundant capacity on the plateau is 1 bit, confirming that the information contained in Alice's fragment of the environment is completely redundant.
The unique capacities for Alice and Bob are both zero on the plateau, consistent with the fact that all environment qubits interact with the target identically.
Only when Alice's (or Bob's) fragment consists of the vast majority of the environment is there unique information.
The synergistic capacity on the plateau then makes up the remainder of the difference between the collective capacity and the (redundant) individual capacity.

\subsection{Non-identical Interactions}
The decoherence model can be generalized to allow for Alice and Bob's environment fragments to interact with the target in non-identical ways.
For instance, while the original model can be thought of as describing photons interacting with the target system, it should be the case that molecules in the air experience a different interaction with the system.
To that end, we define a generalized imperfect C-NOT gate for Alice as
\begin{equation}
    CX_\alpha^{(A)} = \dyad{+_{A}}\otimes\mathbbm{1} + \dyad{-_{A}}\otimes X_\alpha,
\end{equation}
where 
\begin{eqnarray}
    \ket{+_{A}} &=&\cos\frac{\theta_{A}}{2}\ket{0} + e^{i\phi_{A}}\sin\frac{\theta_{A}}{2}\ket{1} \nonumber \\
    \ket{-_{A}} &=&-\sin\frac{\theta_{A}}{2}\ket{0} + e^{i\phi_{A}}\cos\frac{\theta_{A}}{2}\ket{1},
\end{eqnarray}
and $X_\alpha$ is still defined as in (\ref{eq:x-alpha}).
Bob's qubits then interact with the target with an analogous interaction determined by $(\theta_B, \phi_B)$, and the same $\alpha$.
We recover the original model by setting $\theta_A = \theta_B = 2\cos^{-1}\sqrt{\lambda}$ and $\phi_A = \phi_B = 0$.
Without loss of generality, however, we can fix the target qubit to start in the state $\ket{0}$ (i.e.,\ $\lambda=1$).

Alice and Bob's interactions with the target in general will no longer commute, so the order in which they occur matters.
For the purpose of computational tractability, we focus on the extreme case in which all of Alice's qubits interact with the target first.

The superdense capacities for $\theta_A = \pi/4$, $\theta_B = \pi/2$ and $\phi_A = \phi_B = 0$ are depicted in Figure~\ref{fig:GeneralizedZurekCapacities}.

Because Bob's interaction is more effective at encoding information about the target (with the initial state of the target being $\ket{0}$, control qubit states for the C-NOT which lie on the equator of the Bloch sphere optimize fragment/system mutual information), Bob achieves greater superdense capacity on the plateau.
This is also a consequence of Bob's interaction occurring after Alice's, which effectively disturbs the information about the target encoded in Alice's qubits by her interaction.
In fact, due to this order asymmetry, it is impossible for Alice to obtain unique information on the plateau.

With Bob having a greater superdense capacity than Alice on the plateau, it is expected that some of his information about the target should be unique.
Indeed, as depicted in Figure~\ref{fig:GeneralizedZurekCPID}, Bob maintains unique capacity throughout the plateau.
\begin{figure*}
    \subfloat[\label{fig:GeneralizedZurekCapacities}]{%
        \includegraphics[width=\columnwidth]{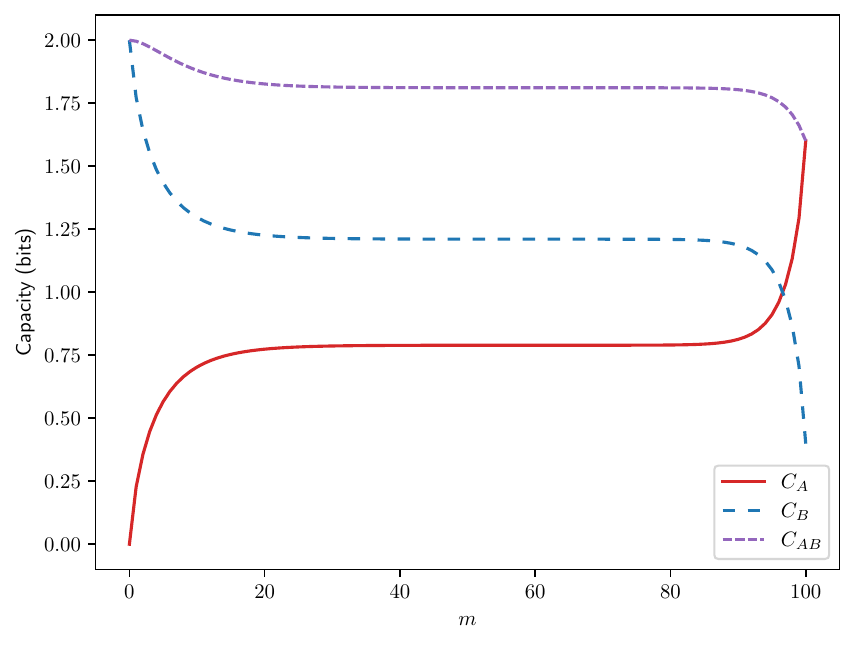}%
    }
    \hspace*{\fill}%
    \subfloat[\label{fig:GeneralizedZurekCPID}]{%
        \includegraphics[width=\columnwidth]{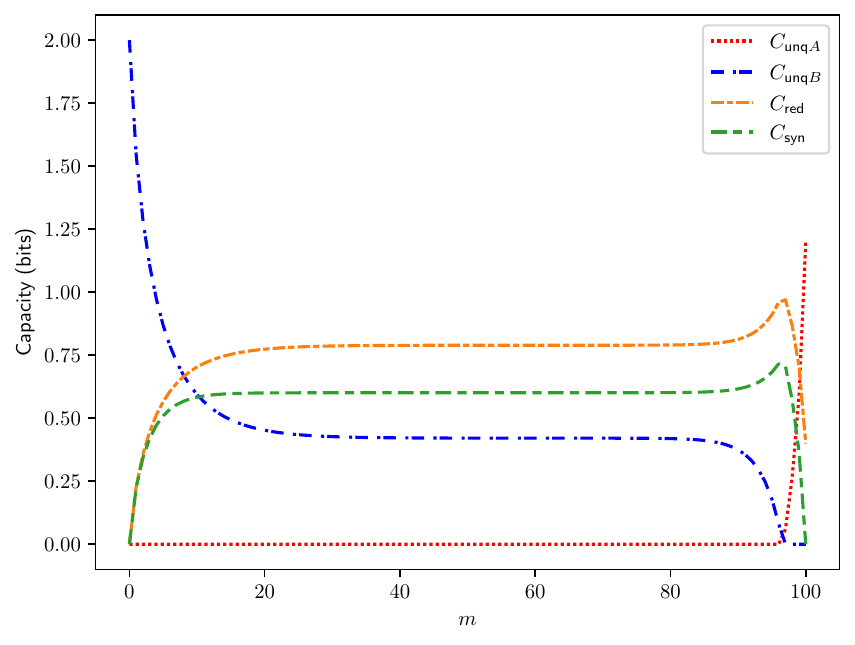}%
    }
    \caption{
        a) Collective and individual superdense capacities for the generalized decoherence model with $\theta_A = \pi/4$, $\theta_B = \pi/2$, $\phi_A = \phi_B = 0$, $\sin\alpha = 0.85$ and $N=100$.
        The classical plateau is still evident in the individual two-bit capacities, but their values are now different for Alice and Bob.
        That Bob's individual capacity on the plateau is greater than Alice's is an unavoidable effect of modeling his interaction as occurring after Alice's.
        Because of the difference in Alice and Bob's interactions with the target, the collective two-bit capacity now depends on the partitioning of the environment qubits. It is maximized when Bob possesses the entirety of the environment qubits, as Bob's interaction is more ``optimal''.
        b) C-PID for the generalized decoherence model.
        Note that Bob now maintains unique information throughout the plateau.
        All of Alice's (and approximately half of Bob's) information is redundant.
        Thus, even with non-identical interactions, information about the target is nonetheless redundantly encoded into the environment consistent with the emergence of classical reality as argued by quantum Darwinists.
    }
\end{figure*}
However, Bob only achieves nonzero unique information on the plateau when his interaction is stronger than Alice's (i.e.,\ $\abs{\theta_B - \pi/2} < \abs{\theta_A - \pi/2}$).
Most notably, despite Alice and Bob interacting with the target in different ways, the largest component of the joint capacity is still redundant.
This further supports the quantum Darwinists' supposition that redundant encoding of information is necessary for the emergence of classical reality from the quantum world.

\section{\label{sec:Summary} Summary and Conclusions}

We generalized the classical notions of unique, redundant, and synergistic information to quantum information. 
We based our definitions on superdense coding. This choice is motivated by two different considerations. 
First, the superdense coding capacity is directly related to the more typically considered mutual information via $C = I + (1-H_T)$, where $H_T$ is the entropy of the reduced state of the target. 
Second, there is an exclusion principle for superdense coding, which states that in a multi-party setting at most one party can possess a quantum advantage (i.e.,\ do better than what would be possible classically). 
This suggests the presence of unique quantum information for the party with the quantum advantage. 
In a PID, defining one of the three concepts suffices to fix the remaining two concepts, and so our definition starts with defining unique quantum superdense coding capacity.

With our definition of a quantum PID in place, we applied it to a handful of different quantum error correction codes. 
More precisely, for the entangled state $\sqrt{\lambda} \ket{0}_T \ket{0}_L +\sqrt{1-\lambda}\ket{1}_T \ket{1}_L$ we considered the information about the target qubit $T$ contained in arbitrary subsets of encoding qubits comprising the logical qubit $L$.
One important feature of this example is that the superdense capacity between $T$ and any subset can take on just 3 different values, $1$ or $1\pm H_T$.

We considered the three-qubit bit-flip code and confirmed that it is classical in a precise sense: any strict subset of the encoding qubits has just one (redundant) bit of information.
The five-qubit code is a perfect code. 
This, too, manifests itself directly in our quantum PID in the combination of two facts: (i) that no subset of encoding qubits ever contains just one (``classical'') redundant bit of capacity, and (ii) that any subset containing the majority of encoding qubits has a quantum advantage and possesses the full $1 + H_T$ bits of capacity. 
The complement of such a subset is correctable and contains only the $1-H_T$ bits of redundant capacity shared by all subsets of the physical qubits.

The seven-qubit Steane code is not perfect and this shows itself in the fact that there are some subsets of 4 encoding qubits that do not contain the full $1+H_T$ bits of superdense capacity. 
On the other hand, there is also no subset with just the one bit of classical capacity, but that is because the Steane code is self-dual.

Shor's nine-qubit code contains quantum aspects but also, being a concatenated version of the three-qubit code, classical aspects in that some subsets of 3 qubits contain just the one classical bit of information about the target qubit. 
This also shows it is not a perfect code.

By studying all these examples we further find (i) erasure-correctable subsets of qubits always have zero unique information, and (ii) that synergistic information between two subsets of encoding qubits is nonzero if and only if there is a logical operation that is supported on the union of the two subsets, but not on the subsets individually.

We also studied the PID for random codes. 
The numerical results show how a little bit of unique information does now exist for subsets of qubits that would be erasure-correctable in an exact quantum code and, similarly, how a little bit of nonzero synergistic information exists where good codes contain none. 
In other words, good error correction codes are special degenerate cases of random codes.

We similarly find that random codes always contain more redundant information than good codes, which confirms that the colloquial way of talking about redundant information in the context of error correction really uses a different notion of ``redundant.'' 
The colloquial meaning is that information is spread out over multiple qubits, but our meaning is more precise and limited: that the same information can be obtained from two disjoint subsets.

Our second application concerns decoherence models, especially as used to illuminate the emergence of the classical world out of a microscopic quantum world. One idea (sometimes called quantum Darwinism) is that it is redundant information that characterizes the classical world. 
If many different parts of the environment possess the same information about a quantum system, then that system has become effectively classical. 
The models to argue for this interpretation so far have been symmetric between the different parts of the environment. 
We modified a simple model and made it asymmetric on purpose to see if (most of) the information in the environment is still redundant. 
This is indeed so, and in this way we confirm the basic premise of quantum Darwinism.

In short, our quantum PID provides additional insights about both quantum error correction codes and decoherence models and conversely, these applications shed light on the idea of a PID.

Finally, we focused here on proposing and testing one definition for a quantum PID. In all our cases the target system was just a qubit. 
An obvious open question is how to extend our definition and calculations to more general cases.

\appendix
\section{\label{sec:Non-Negativity} Non-Negativity of the C-PID}
It is well-known that quantum versions of some information-theoretic quantities, such as conditional entropy, can fail to be non-negative when generalized to the quantum setting.
Such failure of non-negativity can be interpreted as indicating aspects of quantum information that are not present in classical information, such as entanglement.
Our C-PID, however, decomposes the \textit{classical} information capacity of a quantum state.
Non-negativity in this case is desirable, as a negative classical information capacity would have an unclear interpretation.
Hence we prove the following theorem:
\begin{theorem}[C-PID Non-Negativity] \label{thm:C-PID_non-negativity}
    For any tripartite $\rho_{TAB}$, all C-PID elements ($C_{\text{unq}A}$, $C_{\text{unq}B}$, $C_\text{red}$, $C_\text{syn}$) are non-negative.
\end{theorem}

The proof of this theorem requires the following Lemmas:
\begin{Lemma} \label{lem:2_gte_1}
    The superdense capacity is greater than or equal to all one-bit encodings for any subsystem, i.e.,\
    \begin{equation}
        i_S \leq C_S \quad \forall i\in \{X,Y,Z\},\; S \in \{A, B\}.
    \end{equation}
    Or, in terms of the parallel and transverse capacities,
    \begin{equation}
        \mathcal{C}_S^\parallel, \mathcal{C}_S^\perp \leq C_S.
    \end{equation}
\end{Lemma}
\begin{proof}
    The completely depolarizing channel decomposes as
    \begin{equation}
        \Phi_C = \Phi_{i'} \circ \Phi_i
    \end{equation}
    for any Pauli direction $i'$ orthogonal to $i$ (e.g.,\ $\Phi_X \circ \Phi_Z = \Phi_C$). 
    Since $\Phi_{i'}$ is unital, the entropy is non-decreasing under it:
    \begin{equation}
        H(\Phi_C(\rho_{TS})) = H(\Phi_{i'}(\Phi_i(\rho_{TS}))) \ge H(\Phi_i(\rho_{TS})).
    \end{equation}
    Hence,
    \begin{eqnarray}
        C_S &=& H(\Phi_C(\rho_{TS})) - H(\rho_{TS}) \nonumber \\
        &\geq& H(\Phi_i(\rho_{TS})) - H(\rho_{TS}) = i_S
    \end{eqnarray}
\end{proof}

For notational convenience, denote the various encoded states by
\begin{eqnarray} \label{eq:encoded_states}
    \rho_{TS}^i &=& \Phi_i(\rho_{TS}) \quad i \in\{X,Y,Z\} \\
    \rho_{TS}^C &=& \Phi_C(\rho_{TS}),
\end{eqnarray}
and their corresponding entropies by
\begin{eqnarray}
    H_{TS}^i &=& H(\rho_{TS}^i), \\
    H_{TS}^C &=& H(\rho_{TS}^C),
\end{eqnarray}
for $S \subseteq \{A,B\}$.
Note that, as the encoding channels act only on $T$, they have no effect on reduced states with the target traced out, i.e.,\
\begin{equation} \label{eq:marginal_equality}
    \rho_S^C = \rho_S^i = \rho_S.
\end{equation}

\begin{Lemma} \label{lem:joint_gte_ind}
    Joint capacities are greater than or equal to individual capacities, i.e.,\
    \begin{equation}
        i_{AB} \geq i_A, i_B.
    \end{equation}
\end{Lemma}
\begin{proof}
    The difference of one-bit capacities $i_{AB} - i_A$ can be written as a difference of conditional entropies:
    \begin{eqnarray*}
        i_{AB} - i_A &=& (H_{TAB}^i-H_{TAB}) - (H_{TA}^i - H_{TA}) \\
        &=& (H_{TAB}^i - H_{TA}^i) - (H_{TAB} - H_{TA}) \\
        &=& H(B|TA)_{\rho_{TAB}^i} - H(B|TA)_{\rho_{TAB}}
    \end{eqnarray*}
    which is greater than or equal to zero by the Data Processing Inequality~\cite{wilde2013}.
    The same argument applies to the superdense capacity as well.
\end{proof}

\begin{Lemma} \label{lem:diff_of_diff}
    Joint/individual superdense capacity differences are greater than or equal to one-bit capacity differences, i.e.,\
    \begin{equation} \label{eq:superdense_diff_vs_one-bit_diff}
        C_{AB} - C_B \geq i_{A} - i_B.
    \end{equation}
    Or, in terms of the parallel and transverse capacities,
    \begin{equation}
        C_{AB} - C_B \geq \mathcal{C}_A^\parallel - \mathcal{C}_B^\parallel, \; \mathcal{C}_A^\perp - \mathcal{C}_B^\perp 
    \end{equation}
\end{Lemma}
\begin{proof}
    From the strong subadditivity of entropy and~\eqref{eq:marginal_equality} we have that
    \begin{eqnarray}
        H_{TAB}^i &\leq& H_{TB}^i + H_{AB}^i - H_{B}^i \\
        &\leq& H_{TB}^i + H_{AB} - H_{B}. 
    \end{eqnarray}
    Rearranging, adding 1 to each side, and applying Lemma~\ref{lem:joint_gte_ind}
    \begin{eqnarray}
        & 1 + H_{AB} - H_{TAB}^i \;\geq\; 1 + H_B - H_{TB}^i \\
        \implies& C_{AB} - i_{AB} \;\geq\; C_B - i_B \\
        \implies& C_{AB} - i_{A} \;\geq\; C_B - i_B \\
        \implies& C_{AB} - C_B \;\geq\; i_{A} - i_B        
    \end{eqnarray}
\end{proof}

We are now ready to prove the main theorem:
\begin{proof}[Proof of Theorem~\ref{thm:C-PID_non-negativity}]
    Firstly, note that the unique capacities $C_{\text{unq}A}$, $C_{\text{unq}B}$ are manifestly non-negative by construction.
    Next, recall that the C-PID consists of four cases based on the sign and magnitude of $C_A - C_B$ with respect to the one-bit unique capacities $\Gamma_A$, $\Gamma_B$ (see Table~\ref{tab:CPID}).
    
    For cases 1 and 4, the redundant and synergistic capacities are non-negative by the non-negativity of superdense capacities and Lemma~\ref{lem:joint_gte_ind}, respectively.

    Cases 2 and 3 are symmetric under $A \leftrightarrow B$, so it suffices to prove non-negativity of $C_\text{red}$ and $C_\text{syn}$ for only one case.
    The condition for case 3 is
    \begin{equation} \label{eq:case_3_condition}
        \Gamma_A - \Gamma_B \leq C_A - C_B \leq \Gamma_A,
    \end{equation}
    and the redundant and synergistic capacities in this case are given by
    \begin{eqnarray}
        C_\text{red} &=& C_A - \Gamma_A, \label{eq:case_3_red} \\ 
        C_\text{syn} &=& C_{AB} - C_B - \Gamma_A. \label{eq:case_3_syn}
    \end{eqnarray}
    From the definition of the one-bit unique capacity,
    \begin{equation} \label{eq:one-bit-A}
        \Gamma_A  = \max\left(\mathcal{C}_A^\parallel - \mathcal{C}_B^\parallel, 0\right) + \max\left(\mathcal{C}_A^\perp - \mathcal{C}_B^\perp, 0\right),
    \end{equation}
    we see that there are four sub-cases, depending on which of the two terms in~\eqref{eq:one-bit-A} are nonzero.

    When both terms are zero, the equations~\eqref{eq:case_3_red} and~\eqref{eq:case_3_syn} are the same as in C-PID cases 1 and 4 (and hence are non-negative).
    If both terms are nonzero, then $\Gamma_B = 0$, the condition~\eqref{eq:case_3_condition} implies that $\Gamma_A = C_A - C_B$, and the expressions for redundant and synergistic capacities again reduce to those of cases 1 and 4.
    
    Finally, consider the sub-case in which exactly one term is nonzero.
    Assume (without loss of generality) that the parallel term is nonzero, i.e.,\ $\Gamma_A = \mathcal{C}_A^\parallel - \mathcal{C}_B^\parallel$.
    The redundant capacity is then
    \begin{eqnarray}
        C_\text{red} &=& C_A - \left(\mathcal{C}_A^\parallel - \mathcal{C}_B^\parallel\right) \\
        &=& \left(C_A - \mathcal{C}_A^\parallel\right) + \mathcal{C}_B^\parallel.
    \end{eqnarray}
    The term in parentheses is non-negative by Lemma~\ref{lem:2_gte_1}, while the remaining term is non-negative by the non-negativity of capacities.
    Simultaneously, the synergistic capacity is
    \begin{equation}
        C_\text{syn} = \left(C_{AB} - C_B\right) - \left(\mathcal{C}_A^\parallel - \mathcal{C}_B^\parallel\right),
    \end{equation}
    which is non-negative by Lemma~\ref{lem:diff_of_diff}.
\end{proof}

\section{\label{sec:PureStates} Pure States}

When the collective state $\rho_{TAB} = \dyad{\psi_{TAB}}$ is pure, the one-bit unique capacities are in fact suitable definitions for the unique superdense capacities. We prove the following theorem:
\begin{theorem}[Pure-State Unique Capacity] \label{thm:pure_states}
    For pure $\rho_{TAB} = \dyad{\psi_{TAB}}$, at most one of $C_{\text{unq}A}$ and $C_{\text{unq}B}$ is nonzero. In fact, at most one of $\Gamma_A$ and $\Gamma_B$ can be nonzero, and will be equal to $\abs{C_A - C_B}$, so that $C_{\text{unq}A(B)} = \Gamma_{A(B)}$.
\end{theorem}

We first establish the following Lemma:
\begin{Lemma} \label{lem:pure_same_eigenvalues}
    For pure $\rho_{TAB} = \dyad{\psi_{TAB}}$, the one-bit encoded states $\rho_{TA}^i$ and $\rho_{TB}^i$ (see equation~\eqref{eq:encoded_states}) have the same eigenvalues, and hence the same entropy.
\end{Lemma}
\begin{proof}
    We begin by expressing $\rho_{TAB}$ in the Pauli basis,
\[ \rho_{TAB} = \frac{1}{8}\sum_{l,m,n=0}^3 d_{lmn}\sigma_{lmn}, \]
where
\[ \sigma_{lmn} = \sigma_l \otimes \sigma_m \otimes \sigma_n, \]
and
\begin{eqnarray}
  d_{lmn} &=&\Tr[\rho_{TAB}\sigma_{lmn}] \\
  &=&\mel{\psi_{TAB}}{\sigma_{lmn}}{\psi_{TAB}}.
\end{eqnarray}
Note that, for all $\rho_{TAB}$, $d_{000} = 1$.

We define the Pauli matrices for the one-bit encodings with respect to the eigenbasis of $\rho_T$, but this is equivalent to first diagonalizing $\rho_T$ with respect to the standard basis and using the standard Pauli matrices.
The condition that $\rho_T$ be diagonal implies that
\[ d_{100} = d_{200} = 0. \]
Using the identity
\[ \sigma_i\sigma_j\sigma_i = 2\delta_{ij}\sigma_i - \sigma_j \]
for $i,j\in\{1,2,3\}$, we can write the one-bit encoded states as
\begin{eqnarray}
  \rho_{TA}^i &=&\frac{1}{4}\sum_{l=0}^3 \left(d_{0l0}\sigma_{0l0} + d_{il0}\sigma_{il0}\right), \\
  \rho_{TB}^i &=&\frac{1}{4}\sum_{l=0}^3 \left(d_{00l}\sigma_{00l} + d_{i0l}\sigma_{i0l}\right).
\end{eqnarray}

The reduced encoded states $\rho_T^i$ are now diagonal in the eigenbasis of the encoding Pauli $\sigma_i$.
We denote the eigenstates of the Paulis by $\{\ket{+_i}, \ket{-_i}\}$, where
\begin{align}
  \ket{+_1} &= \frac{1}{\sqrt{2}}\left(\ket{0} + \ket{1}\right), &\; \ket{-_1} &= \frac{1}{\sqrt{2}}\left(\ket{0} - \ket{1}\right) \nonumber\\
  \ket{+_2} &= \frac{1}{\sqrt{2}}\left(\ket{0} + i\ket{1}\right), &\; \ket{-_2} &= \frac{1}{\sqrt{2}}\left(\ket{0} - i\ket{1}\right) \nonumber\\
  \ket{+_3} &= \ket{0}, &\; \ket{-_3} &= \ket{1}.
\end{align}
The states $\rho_{TA}^i$ and $\rho_{TB}^i$ are then block-diagonal in the basis $\{\ket{+_i\;0}, \ket{+_i\;1}, \ket{-_i\;0}, \ket{-_i\;1}\}$, i.e.,\
\begin{equation}
    \rho_{TA}^i = \mqty(L^i & 0 \\ 0 & R^i),
\end{equation}
where
\begin{widetext}
\begin{eqnarray}
  L^i &=&\frac{1}{4}\mqty(d_{000} + d_{030} + d_{i30} & d_{010} + d_{i10} -i(d_{020} + d_{i20}) \\ d_{010} + d_{i10} +i(d_{020} + d_{i20}) & d_{000} - d_{030} - d_{i30}) \\
  R^i &=&\frac{1}{4}\mqty(d_{000} + d_{030} - d_{i30} & d_{010} - d_{i10} -i(d_{020} - d_{i20}) \\ d_{010} - d_{i10} +i(d_{020} - d_{i20}) & d_{000} - d_{030} + d_{i30}).
\end{eqnarray}
\end{widetext}
The matrices for $\rho_{TB}^i$ in this basis have the exact same form, just with $d_{lmn} \to d_{lnm}$.
Both submatrices have a trace of 1/2, and their determinants are given by
\begin{widetext}
    \begin{eqnarray}
      \det L^i &=&\frac{1}{16}\big[d_{000}^2 - \left((d_{010} + d_{i10})^2 + (d_{020} + d_{i20})^2 + (d_{030} + d_{i30})^2\right)\big] \nonumber\\
      \det R^i &=&\frac{1}{16}\big[d_{000}^2 - \left((d_{010} - d_{i10})^2 + (d_{020} - d_{i20})^2 + (d_{030} - d_{i30})^2\right)\big].
    \end{eqnarray}
\end{widetext}

Therefore, to show that $\rho_{TA}^i$ and $\rho_{TB}^i$ have the same eigenvalues, it suffices to show that the expression
\[ \sum_{j=1}^3 (d_{0j0} \pm d_{ij0})^2 \]
is invariant under the transposition of the last two indices of the $d$s. We have that
\begin{eqnarray*}
  d_{0j0} \pm d_{ij0} &=& \mel{\psi_{TAB}}{\sigma_{0j0} + \sigma_{ij0}}{\psi_{TAB}} \\
  &=& \mel{\psi_{TAB}}{(\Id \pm \sigma_i)\otimes\sigma_j\otimes\Id}{\psi_{TAB}} \\
  &=& 2\mel{\psi_{TAB}}{\dyad{\pm_i}\otimes\sigma_j\otimes\Id}{\psi_{TAB}}.
\end{eqnarray*}
Now, let $\ket{\phi_\pm^i} = \braket{\pm_i}{\psi_{TAB}}$ be the (unnormalized) state resulting from projecting the $T$ subspace onto $\ket{\pm_i}$.
As $\ket{\phi_\pm^i}$ is a pure state, we can use the Schmidt decomposition to write it as
\[ \ket{\phi_\pm^i} = \sqrt{\lambda_1}\ket{a_1b_1} + \sqrt{\lambda_2}\ket{a_2b_2} \]
for some orthonormal bases $\{\ket{a_1}, \ket{a_2}\}$ and $\{\ket{b_1}, \ket{b_2}\}$.

Now,
\begin{eqnarray}
  \frac{1}{2}(d_{0j0} \pm d_{ij0}) &=&\mel{\phi_\pm^i}{\sigma_j\otimes\Id}{\phi_\pm^i} \nonumber\\
  &=&\left(\lambda_1 \mel{a_1}{\sigma_j}{a_1} + \lambda_2 \mel{a_2}{\sigma_j}{a_2}\right) \nonumber\\
  &=&\mel{a_1}{\sigma_j}{a_1}(\lambda_1 - \lambda_2).
\end{eqnarray}

Finally, then,
\begin{eqnarray}
  \sum_{j=1}^3 (d_{0j0} \pm d_{ij0})^2 &\propto& (\lambda_1 - \lambda_2)^2 \sum_{j=1}^3 \mel{a_1}{\sigma_j}{a_1}^2 \nonumber \\
  &\propto& (\lambda_1 - \lambda_2)^2,
\end{eqnarray}
where the constant of proportionality is determined by the norm of $\ket{\phi_\pm^i}$.
Thus, the determinants of $L^i$ and $R^i$ are invariant under $A \leftrightarrow B$, and hence $\rho_{TA}^i$ and $\rho_{TB}^i$ have the same eigenvalues.
\end{proof}

\begin{proof}[Proof of Theorem~\ref{thm:pure_states}]
    The difference in the individual superdense capacities for pure states is given by
    \begin{eqnarray}
        C_A - C_B &=&\left(1 + H_A - H_{TA}\right) - \left(1 + H_B - H_{TB}\right) \nonumber \\
        &=&\left(H_A - H_B\right) + \left(H_{TB} - H_{TA}\right) \nonumber \\
        &=&2\left(H_A - H_B\right).
    \end{eqnarray}
    The difference between one-bit capacities is, considering e.g.,\ the $X$ encoding,
    \begin{eqnarray}
        X_A - X_B &=&\left(H_{TA}^X - H_{TA}\right) - \left(H_{TB}^X - H_{TB}\right) \nonumber \\
        &=&\left(H_{TB} - H_{TA}\right) + \left(H_{TA}^X - H_{TB}^X\right) \nonumber \\
        &=&\left(H_A - H_B\right) + \left(H_{TA}^X - H_{TB}^X\right),
    \end{eqnarray}
    which, since $H_{TA}^X = H_{TB}^X$ by Lemma~\ref{lem:pure_same_eigenvalues}, is precisely half of the individual superdense capacity difference.
    The same argument applies to the $Y$ and $Z$ encodings.
    Hence, assuming without loss of generality that $H_A > H_B$, the parallel and transverse capacities are
    \begin{eqnarray}
        \mathcal{C}_A^\parallel = \mathcal{C}_A^\perp &=& H_A - H_B, \\
        \mathcal{C}_B^\parallel = \mathcal{C}_B^\perp &=& 0
    \end{eqnarray}
    From which it follows that $\Gamma_A = C_A - C_B$ and $\Gamma_B = 0$.
    Similarly, if $H_A < H_B$, then $\Gamma_B = C_B - C_A$ and $\Gamma_A = 0$.
\end{proof}

\section{\label{sec:ClassicalStates}Quasi-Classical States}

Consider a classical joint probability distribution $p(t,a,b)$ encoded into a tripartite state $\rho_{TAB}$ via
\begin{equation}
    \rho_{TAB} = \sum_{t,a,b} p(t,a,b)\dyad{t,a,b}, \label{eq:classical-state}
\end{equation}
for some choice of orthonormal bases $\{\ket{a}\}$, $\{\ket{b}\}$ for $\mathcal{H}_A$ and $\mathcal{H}_B$.
On such quasi-classical states, all reduced operators are diagonal, the von Neumann entropy reduces to the Shannon entropy, and we write $H_T$, $I(T;A) \equiv I_{TA}$, etc., interchangeably for the quantum and classical quantities.
For these states the C-PID reduces to the well-known minimal mutual information (MMI) PID, as we prove in the following theorem:
\begin{theorem}[Quasi-Classical C-PID] \label{thm:classical_states}
    For a quasi-classical state~\eqref{eq:classical-state}, the C-PID coincides with the MMI PID evaluated on the superdense capacities, and is related to the mutual-information MMI PID by a shift of $(1 - H_T)$ in the redundant information.
\end{theorem}

The proof relies on the following Lemma:
\begin{Lemma} \label{lem:classical_one-bit}
    For a quasi-classical state~\eqref{eq:classical-state}, the parallel one-bit capacity is zero and the transverse one-bit capacity reduces to the superdense capacity,
    \begin{equation}
        \mathcal{C}_A^\parallel = 0, \qquad \mathcal{C}_A^\perp = C_A,
    \end{equation}
    and similarly for $B$.
\end{Lemma}
\begin{proof}
For the parallel capacity, this follows from the fact that $\targetPauli{z}$ and $\rho_{TA}$ are simultaneously diagonal, hence $\Phi_Z(\rho_{TA}) = \frac{1}{2}(\rho_{TA} + \targetPauli{z}\rho_{TA}\targetPauli{z}) = \rho_{TA}$ and $Z_A = 0$.
For the transverse capacity, we note that
\begin{eqnarray*}
        \targetPauli{x}\rho_{TA}\targetPauli{x} &=&\sum_{t,a} p(t,a)\dyad{\bar{t}, a} \\
        &=&\sum_{t,a} p(\bar{t}, a)\dyad{t,a},
\end{eqnarray*}
and so
\begin{eqnarray*}
    \Phi_X(\rho_{TA}) &=& \frac{1}{2}\left(\rho_{TA} + \targetPauli{x}\rho_{TA}\targetPauli{x}\right) \\
    &=& \frac{1}{2} \sum_{t,a}\left(p(t,a)+p(\bar{t},a)\right)\dyad{t,a} \\
    &=& \frac{1}{2} \sum_t \dyad{t} \otimes \sum_a p(a)\dyad{a} \\
    &=& \frac{1}{2} \Id[T] \otimes\rho_A,
\end{eqnarray*}
whence
\begin{eqnarray*}
    X_A &=& H\left(\Phi_X(\rho_{TA})\right) - H(\rho_{TA}) \\
    &=& \left(1 + H(\rho_A)\right) - H(\rho_{TA}) \\
    &=& C_A.
\end{eqnarray*}
The exact same argument applies to the $Y$ capacity, so $X_A = Y_A$ and phase averaging is trivial.
\end{proof}

\begin{proof}[Proof of Theorem~\ref{thm:classical_states}]
By Lemma~\ref{lem:classical_one-bit}, the full C-PID reduces to
\begin{eqnarray}
    C_{\text{unq}A} &=& \max(C_A - C_B, 0) \nonumber \\
    C_{\text{unq}B} &=& \max(C_B - C_A, 0) \nonumber \\
    C_\text{red} &=& \min(C_A, C_B) \nonumber \\
    C_\text{syn} &=& C_{AB} - \max(C_A, C_B).
\end{eqnarray}

This is exactly the well-known ``minimal mutual information'' PID~\cite{Barrett_2015}
\begin{eqnarray}
    I_{\text{unq}A} &=& \max(I(T;A) - I(T;B), 0) \nonumber \\
    I_{\text{unq}B} &=& \max(I(T;B) - I(T;A), 0) \nonumber \\
    I_\text{red} &=& \min(I(T;A), I(T;B)) \nonumber \\
    I_\text{syn} &=& I(T;AB) - \max(I(T;A), I(T;B)),
\end{eqnarray}
with superdense capacity in the place of mutual information.
Considering the relationship between superdense capacity and mutual information,
\begin{equation}
    C_{A} = I(T;A) + (1-H_T),
\end{equation}
we see these PIDs are related simply by a shift in redundant information:
\begin{eqnarray}
    I_{\text{unq}A} &=& C_{\text{unq}A} \nonumber \\
    I_{\text{unq}B} &=& C_{\text{unq}B} \nonumber \\
    I_\text{red} &=& C_\text{red} - (1 - H_T) \nonumber \\
    I_\text{syn} &=& C_\text{syn}
\end{eqnarray}
\end{proof}

\section{\label{sec:StabilizerStates} Superdense Capacity of Stabilizer Subsets}
Consider an $[[n,1,d]]$ stabilizer quantum error correcting code, with $n$ physical qubits encoding one logical qubit.
The codespace $\mathcal{C}$ is defined as the simultaneous +1 eigenspace of all elements of the \textit{stabilizer group} $S$, an abelian subgroup of the $n$-qubit Pauli group:
\begin{equation}
    \mathcal{C} = \{\ket{\psi}\in\mathbbm{C}^{2^n} : s\ket{\psi} = \ket{\psi}\;\forall s \in S \}.
\end{equation}
Anti-commuting logical operators $\bar{Z}$ and $\bar{X}$ are selected from $C(S)\setminus S$ (i.e.,\ elements of the Pauli group that commute with all elements of $S$ but are not themselves in $S$).
The logical basis states are then the $\pm1$ eigenvectors of $\bar{Z}$:
\begin{equation}
    \bar{Z}\ket{0_L} = \ket{0_L}, \; \bar{Z}\ket{1_L}=-\ket{1_L},
\end{equation}
and are the unique states stabilized by the extended stabilizer groups
\begin{eqnarray}
    S_0 &&\coloneqq \langle S, \bar{Z}\rangle, \\
    S_1 &&\coloneqq \langle S, -\bar{Z}\rangle,
\end{eqnarray}
i.e.,\ 
\begin{equation}
    \dyad{i_L} = \prod_{j=1}^{n} \frac{1}{2}(\Id + g_j) = \frac{1}{2^n}\sum_{g \in S_i} g,
\end{equation}
where $g_1, \ldots, g_{n-1}$ are generators of $S$ and $g_n = (-1)^i\bar{Z}$.

Consider a subset $A\subseteq\{1,2,\dots,n\}$ of the encoding qubits with $\abs{A}=m_A$. 
For any $n$-qubit Pauli operator $g = g_A \otimes g_{A^c}$, the partial trace over $A^c$ gives:
\begin{equation}
    \Tr_{A^c}[g] = g_A \cdot \Tr[g_{A^c}] = \begin{cases} 2^{n-m_A}\, g_A & \text{if } g_{A^c} = \Id, \\ 0 & \text{otherwise,} \end{cases}
\end{equation}
since non-identity Pauli operators are traceless.
Define the reduced stabilizer groups on $A$ as
\begin{eqnarray}
    S_i^{A} = \{g_A : g = g_A \otimes \Id[A^c] \in S_i\},
\end{eqnarray}
and denote the reduced density operators of the logical basis states on $A$ as
\begin{eqnarray}
    \rho_A^{\dyad{i_L}} &=& \Tr_{A^c}\dyad{i_L} \\
    &=& \frac{1}{2^{m_A}}\sum_{g_A\in S_i^A}g_A.
\end{eqnarray}
Denote the rank of the reduced stabilizer group by $k_A \coloneqq \log_2\abs{S_i^A}$.
The $\rho_A^{\dyad{i_L}}$ are proportional to projectors onto subspaces of dimension $2^{m_A}/\abs{S_i^A} = 2^{m_A - k_A}$, and hence have entropy
\begin{equation}
    H(\rho_A^{\dyad{i_L}}) = m_A - k_A. \label{eq:reduced_logical_state_ent}
\end{equation}
Note that by purity of $\dyad{i_L}$ we have that $H(\rho_A^{\dyad{i_L}}) = H(\rho_{A^c}^{\dyad{i_L}})$, so
\begin{equation}
    m_A - k_A = (n - m_A) - k_{A^c}. \label{eq:reduced_logical_state_ent_relation}
\end{equation}
Note also that if $\bar{Z}$ has a coset representative supported on $A$ then $\rho_A^0$ and $\rho_A^1$ have orthogonal support, otherwise they are equal.

For the encoded state~\eqref{eq:encoded_state}, the reduced state on $A$ is
\begin{equation}
    \rho_A = \lambda\rho_A^{\dyad{0_L}} + (1-\lambda)\rho_A^{\dyad{1_L}}
\end{equation}
If $\bar{Z}$ does not have a representative supported on $A$, then $\rho_A = \rho_A^0 = \rho_A^1$ and the entropy of the reduced state is given by~\eqref{eq:reduced_logical_state_ent}.
Otherwise, $\rho_A$ is a convex combination of two projectors with orthogonal support and therefore has entropy
\begin{eqnarray*}
    H(\rho_A) &=& h(\lambda) + \lambda H(\rho_A^0) + (1-\lambda)H(\rho_A^1) \\
    &=& H_T + m_A - k_A.
\end{eqnarray*}
By defining the accessibility indicator
\begin{equation}
    \epsilon_A = \begin{cases}
        1 & \bar{Z} \text{ supported on } A \\
        0 & \text{otherwise},
    \end{cases}
\end{equation}
we can write the entropy of the reduced state $\rho_A$ succinctly as
\begin{equation}
    H(\rho_A) = m_A - k_A + \epsilon_A H_T \label{eq:reduced_encoded_state_ent}
\end{equation}

From the purity of the full state~\eqref{eq:encoded_state}, the superdense capacity for the subset $A$ can be written in terms of the entropy of the reduced states on $A$ and $A^c$:
\begin{eqnarray}
    C_A &=& 1 + H(\rho_A) - H(\rho_{TA}) \\
    &=& 1 + H(\rho_A) - H(\rho_{A^c}).
\end{eqnarray}
Using~\eqref{eq:reduced_encoded_state_ent} and~\eqref{eq:reduced_logical_state_ent_relation}, this can then be written as
\begin{equation}
    C_A = 1 + \left(\epsilon_A - \epsilon_{A^c}\right)H_T, \label{eq:superdense_accessibility_formula}
\end{equation}
i.e.,\ the superdense capacity depends on whether or not $\bar{Z}$ has a coset representative supported on $A$ and/or its complement. 

Denote by $l_A$ the number of independent logical operator cosets supported on $A$.
The number of cosets supported on a subset and on its complement is governed by the \textit{cleaning lemma}:
\begin{Lemma}[Cleaning Lemma~\cite{haah2013latticequantumcodesexotic}] \label{lem:cleaning}
    For an $[[n,k,d]]$ stabilizer code, the number of independent logical operators supported on a subset is related to the number supported on its complement by
    \begin{equation} \label{eq:cleaning_Lemma}
        l_A + l_{A^c} = 2k.
    \end{equation}
\end{Lemma}
For $k=1$, it can further be shown that in the $l_A = l_{A^c} = 1$ case it is in fact the \textit{same} logical operator coset that is supported on $A$ and $A^c$.

The superdense capacity of a subset $A$ can then be expressed as a function of $l_A$ according to the following theorem:
\begin{theorem} \label{thm:main}
    For the state $\ket{\Psi} = \alpha\ket{0}\ket{0_L} + \beta\ket{1}\ket{1_L}$, the superdense capacity between the unencoded qubit and any subset $A$ of the physical qubits,
    \begin{equation}
        \label{eq:main}
        C_A = 1 + (l_A - 1)\, H_T,
    \end{equation}
    or equivalently,
    \begin{equation} \label{eq:superdense_capacity_logical_count_cases}
        C_A = \begin{cases}
            1 + H_T & \text{if } l_A = 2, \\
            1       & \text{if } l_A = 1, \\
            1 - H_T & \text{if } l_A = 0.
        \end{cases}
    \end{equation}
\end{theorem}
\begin{proof}
    The three cases of~\eqref{eq:superdense_capacity_logical_count_cases} follow from~\eqref{eq:superdense_accessibility_formula}.
    If $l_A = 2$, every logical operator has a representative on $A$, so $\epsilon_A = 1$; by~\eqref{eq:cleaning_Lemma} $l_{A^c} = 0$, so $\epsilon_{A^c} = 0$, and~\eqref{eq:superdense_accessibility_formula} gives $C_A = 1 + H_T$.
    The case $l_A = 0$ is symmetric, giving $C_A = 1 - H_T$.
    For $l_A = 1$, both $A$ and $A^c$ support the same single logical operator.
    Regardless of whether or not that operator is $\bar{Z}$, in this case $\epsilon_A = \epsilon_{A^c}$ and $C_A = 1$.
\end{proof}

\bibliography{refs}

\end{document}